\documentclass[draft, onecolumn, 11pt]{IEEEtran}
\usepackage{cite}
\usepackage{graphicx}
\usepackage{amsmath,amssymb}
\usepackage{dsfont}
\usepackage{algorithmic,algorithm}
\usepackage{epsfig}
\psfull

\newtheorem{thm}{Theorem}[section]

\newtheorem{lem}{Lemma}[section]

\begin{document}
\title{Location-Aided Fast Distributed Consensus in Wireless Networks}

\author{Wenjun~Li, Yanbing Zhang, and~Huaiyu~Dai*, \IEEEmembership{Member}
\thanks{This
research was supported in part by the National Science Foundation
under Grant CCF-0515164, CNS-0721815, and CCF-0830462.}
\thanks{W. Li is with Qualcomm Inc, San Diego, CA 92121 (e-mail:
wenjunl@qualcomm.com). The work was done when she was with NC State
University.}
\thanks{Y. Zhang and H. Dai are with the ECE department of NC State University,
Raleigh, NC 27695 (e-mail: \{yzhang,huaiyu\_dai\}@ncsu.edu).}
\date{}
}\maketitle \markboth{Submitted to IEEE Trans. Inform. Theory.}{}
\begin{abstract}
Existing works on distributed consensus explore linear iterations
based on \emph{reversible} Markov chains, which contribute to the
slow convergence of the algorithms. It has been observed that by
overcoming the diffusive behavior of reversible chains, certain
nonreversible chains lifted from reversible ones mix substantially
faster than the original chains. In this paper, we investigate the
idea of accelerating distributed consensus via lifting Markov
chains, and propose a class of Location-Aided Distributed Averaging
(LADA) algorithms for wireless networks, where nodes' coarse
location information is used to construct nonreversible chains that
facilitate distributed computing and cooperative processing. First,
two general pseudo-algorithms are presented to illustrate the notion
of distributed averaging through chain-lifting. These
pseudo-algorithms are then respectively instantiated through one
LADA algorithm on grid networks, and one on general wireless
networks. For a $k\times k$ grid network, the proposed LADA
algorithm achieves an $\epsilon$-averaging time of
$O(k\log(\epsilon^{-1}))$. Based on this algorithm, in a wireless
network with transmission range $r$, an $\epsilon$-averaging time of
$O(r^{-1}\log(\epsilon^{-1}))$ can be attained through a centralized
algorithm. Subsequently, we present a fully-distributed LADA
algorithm for wireless networks, which utilizes only the direction
information of neighbors to construct nonreversible
chains. 
It is shown that this distributed LADA algorithm achieves the same
scaling law in averaging time as the centralized scheme in wireless
networks for all $r$ satisfying the connectivity requirement. The
constructed chain attains the optimal scaling law in terms of an
important mixing metric, the fill time, among all chains lifted from
one with an approximately uniform stationary distribution on
geometric random graphs. Finally, we propose a cluster-based LADA
(C-LADA) algorithm, which, requiring no central coordination,
provides the additional benefit of reduced message complexity
compared with the distributed LADA algorithm.
\end{abstract}


\section{Introduction}
As a basic building block for networked information processing,
distributed consensus admits many important applications in various
areas, such as distributed estimation and data fusion, coordination
and cooperation of autonomous agents, as well as network
optimization. The distributed averaging problem where nodes try to
reach consensus on the average value\footnote{With appropriate
modification, such algorithms can also be extended to computation of
weighted sums, linear synopses, histograms and types, and can
address a large class of distributed computing and statistical
inferencing problems.} through iterative local information exchange
has been vigorously investigated recently \cite{Xiao, Bertsekas,
Blondel, Boyd_INFOCOM, Boyd_TIT, Moalleimi}. Compared with
centralized counterparts, such distributed algorithms scale well as
the network grows, and exhibit robustness to node and link failures.
Distributed consensus can be realized through linear iteration in
the form $\mathbf{x}(t+1)=\mathbf{W}(t)\mathbf{x}(t)$ where
$\mathbf{W}(t)$ is a graph conformant matrix\footnote{For a graph
$G=(V,E)$ with the vertex set $V$ and edge set $E$, a matrix
$\mathbf{W}$ of size $|V|\times|V|$ is $G$-conformant, if
$W_{ij}\neq 0$ only if $(i,j)\in E$.}. Distributed averaging through
linear iteration with a deterministic $\mathbf{W}$ is studied in
\cite{Xiao}. For time-varying $\mathbf{W}(t)$, convergence is
guaranteed under mild conditions\cite{Bertsekas, Blondel}. The class
of randomized gossip algorithms recently studied by Boyd \emph{et
al}\cite{Boyd_INFOCOM, Boyd_TIT} realizes consensus through
iterative pairwise averaging, and allows for asynchronous operation.
In their study, independent and identically distributed random
matrices $\mathbf{W}(t)$ are considered, and performance of the
proposed algorithms is governed by the second largest eigenvalue of
$\mathds{E}\left[\mathbf{W}(t)\right]$.

Typically, governing matrices in distributed consensus algorithms
are chosen to be stochastic, which connects them closely to Markov
chain theory. It is also convenient to view the evolvement of a
Markov chain $\mathbf{P}$ as a random walk on a graph (with vertex
set $V$ being the state space of the chain, and edge set $E=\{uv:
P_{uv}>0\}$). In both fixed and random algorithms studied in
\cite{Xiao, Boyd_INFOCOM, Boyd_TIT}, mainly a symmetric, doubly
stochastic weight matrix is used, hence the convergence time of such
algorithms is closely related to the mixing time of a reversible
random walk, which is usually slow due to its diffusive behavior. It
has been shown in \cite{Boyd_TIT} that in a wireless network of size
$n$ with a common transmission range $r$, the optimal gossip
algorithm requires $\Theta\left(r^{-2}\log
(\epsilon^{-1})\right)$\footnote{We use the following order
notations in this paper: Let $f(n)$ and $g(n)$ be nonnegative
functions for $n\geq 0$. We say $f(n) = O(g(n))$ and
$g(n)=\Omega(f(n))$ if there exists some $k$ and $c>0$, such that
$f(n)\leq cg(n)$ for $n\geq k$; $f(n) = \Theta(g(n))$ if $f(n) =
O(g(n))$ as well as $f(n) = \Omega(g(n))$. We also say $f(n) =
o(g(n))$ and $g(n) = \omega(f(n))$ if
$\lim_{n\rightarrow\infty}\frac{f(n)}{g(n)}=0$.} time for the
relative error to be bounded by $\epsilon$. This means that for a
small radius of transmission, even the fastest gossip algorithm
converges slowly.

Reversible Markov chains are dominant in research literature, as
they are mathematically more tractable -- see \cite{Aldous} and
references therein. However, it is observed by Diaconis \emph{et
al.}\cite{Diaconis} and later by Chen \emph{et al.} \cite{Chen} that
certain nonreversible chains mix substantially faster than
corresponding reversible chains, by
overcoming the diffusive behavior of reversible random walks. 
Our work is directly motivated by this finding, as well as the close
relationship between distributed consensus algorithms and Markov
chains. We first show that by allowing each node in a network to
maintain multiple values, mimicking the multiple lifted states from
a single state, a nonreversible chain on a lifted state space can be
simulated, and we present two general pseudo-algorithms for this
purpose. The next and more challenging step is to explicitly
construct fast-mixing non-reversible chains given the network
graphs. In this work, we propose a class of Location-Aided
Distributed Averaging (LADA) algorithms that result in significantly
improved averaging times compared with existing algorithms. As the
name implies, the algorithms utilize (coarse) location information
to construct nonreversible chains that prevent the same information
being ``bounced" forth and back, thus accelerating information
dissemination.

Two important types of networks, grid networks and general wireless
networks modeled by geometric random graphs, are considered in this
work. For a $k\times k$ grid, we propose a LADA algorithm as an
application of our Pseudo-Algorithm 1, and show that it takes
$O(k\log(\epsilon^{-1}))$ time to reach a relative error within
$\epsilon$. Then, for the celebrated geometric random graph $G(n,
r)$ with a common transmission range $r$, we present a centralized
grid-based algorithm which exploits the LADA algorithm on the grid
to achieve an $\epsilon$-averaging time of
$O(r^{-1}\log(\epsilon^{-1}))$.

In practice, purely distributed algorithms requiring no central
coordination are typically preferred. Consequently, we propose a
fully-distributed LADA algorithm, as an instantiation of
Pseudo-Algorithm 2.
On a wireless network with randomly distributed nodes, the
constructed chain does not possess a uniform stationary distribution
desirable for distributed averaging, due to the difference in the
number of neighbors a node has in different directions.
Nevertheless, we show that the non-uniformity for the stationary
distribution can be compensated by weight variables which estimate
the stationary probabilities, and that the algorithm achieves an
$\epsilon$-averaging time of $O(r^{-1}\log (\epsilon^{-1}))$ with
any transmission range $r$ guaranteeing network connectivity.
Although it is not known whether the achieved averaging time is
optimal for all $\epsilon$, we demonstrate that the constructed
chain does attain the optimal scaling law in terms of another mixing
metric $T_{\mathrm{fill}}(\mathbf{P},c)$ (c.f. (\ref{fill})), among
all chains lifted from one with an approximately (on the order
sense) uniform stationary distribution on $G(n,r)$. In Appendix
\ref{appLADAU}, we provide another algorithm, the LADA-U algorithm,
where the nonreversible chain is carefully designed to ensure an
exact uniform stationary distribution (which accounts for the suffix
``U"), by allowing some controlled diffusive behavior. It is shown
that LADA-U can achieve the same scaling law in averaging time as
the centralized and distributed LADA algorithm, but needs a larger
transmission range than minimum connectivity requirement, mainly due
to the induced diffusive behavior.

Finally, we propose a cluster-based LADA (C-LADA) variant to further
improve on the message complexity. This is motivated by the common
assumption that nodes in some networks, such as wireless sensor
networks, are densely deployed, where it is often more efficient to
have co-located nodes clustered, effectively behaving as a single
entity. In this scenario, after initiation, only inter-cluster
communication and intra-cluster broadcast are needed to update the
values of all nodes. Different from the centralized algorithm,
clustering is performed through a distributed clustering algorithm;
the induced graph is usually not a grid, so the distributed LADA
algorithm, rather than the grid-based one, is suitably modified and
applied. The same time complexity as LADA is achieved, but the
number of messages per iteration is reduced from $\Theta(n)$ to
$\Theta(r^{-2})$.

In this paper, for ease of exposition we focus on synchronous
algorithms without gossip constraints, i.e., in each time slot,
every node updates its values based on its neighbors' values in the
previous iteration. Nonetheless, these algorithms can also be
realized in a deterministic gossip fashion, by simulating at most
$d_{\max}$ matchings for each iteration, where $d_{\max}$ is the
maximum node degree. Also note that while most of our analysis is
conducted on the geometric random graph, the algorithms themselves
can generally be applied on any network topology.

Our paper is organized as follows. In Section II, we formulate the
problem and review some important results in Markov chain theory. In
Section III, we introduce the notion of lifting Markov chains and
present two pseudo-algorithms for distributed consensus based on
chain-lifting. In Section IV, the LADA algorithm for grid networks
is proposed, which is then extended to a centralized algorithm for
geometric random graphs. In Section V, we present the distributed
LADA algorithm for wireless networks and analyze its performance.
The C-LADA algorithm is treated in Section VI. Several important
related works are discussed in Section VII. Finally, conclusions are
given in Section VIII.

\section{Problem Formulation and Preliminaries}
\subsection{Problem Formulation}
Consider a network represented by a connected graph $G=(V,E)$, where
the vertex set $V$ contains $n$ nodes and $E$ is the edge set. Let
vector $\mathbf{x}(0)=[x_1(0), \cdots, x_n(0)]^T$ contain the
initial values observed by the nodes, and
$x_{\mathrm{ave}}=\frac{1}{n}\sum_{i=1}^n x_i$ denote the average.
The goal is to compute $x_{\mathrm{ave}}$ in a distributed and
robust fashion. As we mentioned, such designs are basic building
blocks for distributed and cooperative information processing in
wireless networks. Let $\mathbf{x}(t)$ be the vector containing node
values at the $t$th iteration. Without loss of generality, we
consider the set of initial values $\mathbf{x}(0)\in
{\mathbb{R}^+}^n$, and define the $\epsilon$-averaging time as
\begin{equation}\label{avetime}
T_{\mathrm{ave}}(\epsilon)=\sup_{\mathbf{x}(0)\in
{\mathbb{R}^+}^n}\inf\left\{t:\|\mathbf{x}(t)-x_{\mathrm{ave}}\mathbf{1}\|_1\leq
\epsilon\|\mathbf{x}(0)\|_1\right\}\footnote{For the more general
case $\mathbf{x}(0)\in {\mathbb{R}}^n$, the corresponding expression
in (\ref{avetime}) is
$\|\mathbf{x}(t)-x_{\mathrm{ave}}\mathbf{1}\|_1\leq
\epsilon\|\mathbf{x}(0)-\min_i x_i(0)\mathbf{1}\|_1$.}
\end{equation}
where  $\|\mathbf{x}\|_1=\sum_i|x_i|$ is the $l_1$ norm\footnote{In
the literature of distributed consensus, the $l_2$ norm
$\|\mathbf{x}\|_2=\sqrt{\sum_i|x_i|^2}$ has also been used in
measuring the averaging time\cite{Xiao, Boyd_TIT}. The two metrics
are closely related. Define
$T_{\mathrm{ave},2}(\epsilon)=\sup_{\mathbf{x}(0)\in
{\mathbb{R}^+}^n}\inf\left\{t:\|\mathbf{x}(t)-x_{\mathrm{ave}}\mathbf{1}\|_2\leq
\epsilon\|\mathbf{x}(0)\|_2\right\}$. It is not difficult to show
that when $\epsilon=O\left(\frac{1}{n}\right)$, then
$T_{\mathrm{ave},2}(\epsilon) =
O\left(T_{\mathrm{ave}}(\epsilon)\right)$.}.

We will mainly use the geometric random graph \cite{Gupta,Penrose}
to model a wireless network in our analysis. In the geometric random
graph $G(n,r(n))$, $n$ nodes are uniformly and independently
distributed on a unit square $[0,1]^2$, and $r(n)$ is the common
transmission range of all nodes. It is known that the choice of
$r(n)\geq\sqrt{\frac{2\log n}{n}}$ is required to ensure the graph
is connected with high probability (w.h.p.)\footnote{with
probability approaching 1 as
$n\rightarrow\infty$}\cite{Gupta,Penrose}.

\subsection{Markov Chain Preliminaries}
%
The averaging time of consensus algorithms evolving according to a
stationary Markov chain is closely related to the chain's
convergence time. In this section, we briefly review two metrics
that characterize the convergence time of a Markov chain, i.e., the
mixing time and the fill time. For $\epsilon>0$, the
$\epsilon$-mixing time of an irreducible and aperiodic Markov chain
$\mathbf{P}$ with stationary distribution $\mbox{\boldmath$\pi$}$ is
defined in terms of the total variation distance as\cite{Aldous}
\begin{equation}\label{mixing}
T_{\mathrm{mix}}(\mathbf{P},
\epsilon)\triangleq\sup_{i}\inf\left\{t:\|\mathbf{P}^t(i,\cdot)-\mbox{\boldmath$\pi$}\|_{TV}\triangleq\frac{1}{2}\|\mathbf{P}^t(i,\cdot)-\mbox{\boldmath$\pi$}\|_1\leq
\epsilon\right\}=\sup_{\mathbf{p}(0)}\inf\left\{t:
\|\mathbf{p}(t)-\mbox{\boldmath$\pi$}\|_1\leq 2\epsilon\right\},
\end{equation}
where $\mathbf{p}(t)$ is the probability distribution of the chain
at time $t$, and $\mathbf{P}^t(i,\cdot)$ is the $i$th row of the
$t$-step transition matrix (i.e., $\mathbf{p}(t)$ given
$\mathbf{p}(0)=\mathbf{e}_i^T$\footnote{$\mathbf{e}_i$ is the vector
with 1 at the $i$th position and 0 elsewhere.}). The second equality
is due to the convexity of the $l_1$ norm.


Another related metric, known as the fill time \cite{Lovasz} (or the
separate time \cite{Aldous3}), is defined for $0<c<1$ as
\begin{eqnarray}\label{fill}
T_{\mathrm{fill}}(\mathbf{P},c)\triangleq\sup_{i}\inf\left\{t:\mathbf{P}^t(i,\cdot)>
(1-c)\mbox{\boldmath$\pi$}\right\}.
\end{eqnarray}

For certain Markov chains, it is (relatively) easier to obtain an
estimate for $T_\mathrm{fill}$ than for $T_{\mathrm{mix}}$. The
following lemma comes handy in establishing an upper bound for the
mixing time in terms of $T_\mathrm{fill}$, and will be used in our
analysis.

\begin{lem}\label{fillmix} For any irreducible and aperiodic Markov chain $\mathbf{P}$,
\begin{eqnarray}
T_{\mathrm{mix}}(\mathbf{P},\epsilon)\leq
\left[\log(\epsilon^{-1})/\log(c^{-1})+1\right]T_{\mathrm{fill}}(\mathbf{P},c).
\end{eqnarray}
\end{lem}

\begin{proof}
The lemma follows directly from a well-known result in Markov chain
theory (see the fundamental theorem in Section 3.3 of \cite{Neal}).
It states that for a stationary Markov chain $\mathbf{P}$ on a
finite state space with a stationary distribution
$\mbox{\boldmath$\pi$}$, if there exists a constant $0<c<1$ such
that $P(i,j)>(1-c)\pi_j$ for all $i,j$, then the distribution of the
chain at time $t$ can be expressed as a mixture of the stationary
distribution and another arbitrary distribution $\mathbf{r}(t)$ as
\begin{eqnarray}
\mathbf{p}(t)=(1-c^t)\mbox{\boldmath$\pi$}+c^t \mathbf{r}(t).
\end{eqnarray}
Thus
\begin{eqnarray}
\|\mathbf{p}(t)-\mbox{\boldmath$\pi$}\|_1=c^t\|\mbox{\boldmath$\pi$}-\mathbf{r}(t)\|_1\leq 2c^t.
\end{eqnarray}
Now, for any irreducible and aperiodic chain, by (\ref{fill}), we
have $P^{\tau}(i,j)>(1-c)\pi_j$ for any $i,j$ when
$\tau>T_{\mathrm{fill}}(\mathbf{P}, c)$. It follows from the above
that for any starting distribution,
\begin{eqnarray}\label{eq7}
\frac{1}{2}\|\mathbf{p}(t)-\mbox{\boldmath$\pi$}\|_1\leq
c^{\llcorner t/T_{\mathrm{fill}}(\mathbf{P},~c)\lrcorner},
\end{eqnarray}
and the desired result follows immediately by equating the right
hand side of (\ref{eq7}) with $\epsilon$.
\end{proof}

\section{Fast Distributed Consensus Via Lifting Markov Chains}
The idea of the Markov chain lifting was first investigated in
\cite{Diaconis, Chen} to accelerate convergence. A lifted chain is
constructed by creating multiple replica states corresponding to
each state in the original chain, such that the transition
probabilities and stationary probabilities of the new chain conform
to those of the original chain. Formally, for a given Markov chain
$\mathbf{P}$ defined on state space $V$ with stationary
probabilities $\mbox{\boldmath$\pi$}$, a chain $\mathbf{\tilde{P}}$
defined on state space $\tilde{V}$ with stationary probability
$\tilde{\mbox{\boldmath$\pi$}}$ is a lifted chain of $\mathbf{P}$ if
there is a mapping $f: \tilde{V}\rightarrow V$ such that
\begin{eqnarray}\label{lifting1}
\pi_v = \sum_{\tilde{v}\in f^{-1}(v)} \tilde{\pi}_{\tilde{v}}, \quad
\forall v\in V
\end{eqnarray}
and
\begin{eqnarray}
P_{uv} = \sum_{\tilde{u}\in f^{-1}(u), \tilde{v}\in f^{-1}(v)}
\frac{\tilde{\pi}_{\tilde{u}}}{\pi_u}\tilde{P}_{\tilde{u}\tilde{v}},
\quad \forall u, v\in V.
\end{eqnarray}
Moreover, $\mathbf{P}$ is called a collapsed chain of $\mathbf{\tilde{P}}$.

Given the close relationship between Markov chains and distributed
consensus algorithms, it is natural to ask whether the nonreversible
chain-lifting technique could be used to speed up distributed
consensus in wireless networks. We answer the above question in two
steps.  First, we show that by allowing each node to maintain
multiple values, mimicking the multiple lifted states from a single
state, a nonreversible chain on a lifted state space can be
simulated\footnote{Although sometimes used interchangeably in
related works, in this study it is better to differentiate between
nodes (in a network) and states (in a Markov chain), since several
states in the lifted chain correspond to a single node in a
network.}. In this section, we provide two pseudo-algorithms to
illustrate this idea. With such pseudo-algorithms in place, the
second step is to explicitly construct fast-mixing non-reversible
chains that result in improved averaging times compared with
existing algorithms. The latter part will be treated in Section IV
and V, where we provide detailed algorithms for both grid networks
as well as general wireless networks modeled by geometric random
graphs.

Consider a wireless network modeled as $G(V ,E)$ with $|V|=n$. A
procedure that realizes averaging through chain-lifting is given in
Pseudo-algorithm 1, where $\mathbf{P}$ is some $G$-conformant
ergodic chain on $V$ with a uniform stationary distribution.
\begin{algorithm}\caption{Pseudo-Algorithm 1.}
\begin{enumerate}
\item Each node $v\in V$
maintains $b_v$ copies of values $y_v^1, \cdots, y_v^{b_v}$, the sum
of which is initially set equal to $x_v(0)$. Correspondingly, we
obtain a new state space $\tilde{V}$ and a mapping $f:
\tilde{V}\rightarrow V$ with the understanding that
$\{y_v^l\}_{l=1,\cdots,b_v}$ can be alternatively represented as
$\{y_{\tilde{v}}\}_{\tilde{v} \in f^{-1}(v)}$.
\item At each time instant $t$, each node updates its values based
on the values of its neighbors. Let the vector $\mathbf{y}$ contain
the copies of values of all nodes, i.e.,
$\mathbf{y}=[\mathbf{y}_1^T, \cdots, \mathbf{y}_{|V|}^T]^T$ with
$\mathbf{y}_v=[y_v^1, \cdots, y_v^{b_v}]^T$. The values are updated
according to the linear iteration $\mathbf{y}(t+1) =
\mathbf{\tilde{P}}^T\mathbf{y}(t)$, where $\mathbf{\tilde{P}}$ is
some ergodic chain on $\tilde{V}$ lifted from $\mathbf{P}$.
\item At each time instant $t$, each node estimates the average value
by summing up all its copies of values: $x_v(t)=\sum_{l=1}^{b_v} y_v^{l}(t)$.
\end{enumerate}
\end{algorithm}

\begin{lem}\label{lifted-uniform}
Using Pseudo-algorithm 1, $\mathbf{x}(t)\rightarrow x_{\mathrm{ave}}\mathbf{1}$ and the averaging time
$T_{\mathrm{ave}}(\epsilon)\leq T_{\mathrm{mix}}(\mathbf{\tilde{P}}, \epsilon/2)$.
\end{lem}
\begin{proof}
Let $\mathbf{\tilde{p}}(t)$ be the distribution of
$\mathbf{\tilde{P}}$ at time $t$, and
$\tilde{\mbox{\boldmath$\pi$}}$ the stationary distribution of
$\mathbf{\tilde{P}}$. As $\mathbf{\tilde{P}}$ is ergodic and the
linear iteration in Pseudo-algorithm 1 is sum-preserving, it can be
shown that $\mathbf{y}(t)\rightarrow
nx_{\mathrm{ave}}\tilde{\mbox{\boldmath$\pi$}}$, and
$\mathbf{x}(t)\rightarrow x_{\mathrm{ave}}\mathbf{1}$ due to the
lifting property (\ref{lifting1}) and the uniform stationary
distribution of $\mathbf{P}$. Furthermore, we have $\mathbf{y}(t) =
nx_{\mathrm{ave}}\mathbf{\tilde{p}}(t)$, and for $t\geq
T_{\mathrm{mix}}(\mathbf{\tilde{P}}, \epsilon/2)$,
\begin{eqnarray*}
&&\|\mathbf{x}(t)-x_{\mathrm{ave}}\mathbf{1}\|_1=\sum_{v\in V}|x_v(t)-x_{\mathrm{ave}}|=\sum_{v\in V}|\sum_{l=1}^{b_v}y_v^l-x_{\mathrm{ave}}|=\sum_{v\in V}|\sum_{\tilde{v}\in f^{-1}(v)}(y_{\tilde{v}}(t)-\tilde{\pi}_{\tilde{v}} nx_{\mathrm{ave}})|\\
&\leq& \sum_{v\in V}\sum_{\tilde{v}\in
f^{-1}(v)}|y_{\tilde{v}}(t)-\tilde{\pi}_{\tilde{v}}
nx_{\mathrm{ave}}| =
nx_{\mathrm{ave}}\sum_{\tilde{v}\in\tilde{V}}|\tilde{p}_{\tilde{v}}(t)-\tilde{\pi}_{\tilde{v}}|
\leq nx_{\mathrm{ave}}\epsilon=\epsilon\|\mathbf{x}(0)\|_1,
\end{eqnarray*}
where the third equality is by $\pi_v = \sum_{\tilde{v}\in
f^{-1}(v)} \tilde{\pi}_{\tilde{v}} =\frac{1}{n}$,  $\forall v\in V$,
the first inequality is by the triangular inequality, and the last
inequality is by the definition of mixing time in (\ref{mixing}).
\end{proof}

From the above discussion, we see that for a wireless network
modeled as $G=(V,E)$, as long as we can find a fast-mixing chain
whose collapsed chain is $G$ conformant and has a uniform stationary
distribution on $V$, we automatically obtain a fast distributed
averaging algorithm on $G$. The crux is then to design such lifted
chains which are typically nonreversible to ensure fast-mixing.
While the fact that the collapsed Markov chain possesses a uniform
stationary distribution facilitates distributed consensus, this does
not preclude the possibility of achieving consensus by lifting
chains with non-uniform stationary distributions. In fact, the
non-uniformity of stationary distribution can be ``smoothen out" by
incorporating some auxiliary variables that asymptotically estimate
the stationary distribution. Such a procedure allows us more
flexibilities in finding a fast-mixing chain on a given graph. This
idea is presented in Pseudo-algorithm 2, where $\mathbf{P}$ is some
$G$-conformant ergodic chain on $V$.

\begin{algorithm}\caption{Pseudo-Algorithm 2.}
\begin{enumerate}
\item Each node $v\in V$ maintains $b_v$ pairs of values $(y_v^{l},
w_v^{l})$, $l=1,\cdots b_v$, whose initial values satisfy $\sum_l
y_v^{l}(0) =x_v(0)$ and $\sum_l w_v^{l}(0) = 1$. Correspondingly, we
obtain a new state space $\tilde{V}$ and a mapping $f:
\tilde{V}\rightarrow V$.
\item Let the vector $\mathbf{y}$ contain
the copies $y_v^{l_v}$ for all $v\in V$ and $l_v=1,\cdots, b_v$, and
similarly denote $\mathbf{w}$. At each time instant, the values are
updated with
     \begin{eqnarray*}
     &&\mathbf{y}(t+1) = \mathbf{\tilde{P}}^T\mathbf{y}(t),\\
     &&\mathbf{w}(t+1) = \mathbf{\tilde{P}}^T\mathbf{w}(t),
\end{eqnarray*}
where $\mathbf{\tilde{P}}$ is some ergodic chain on $\tilde{V}$
lifted from $\mathbf{P}$.
\item At each time instant, each node estimates the average
value by
 \begin{eqnarray*}
 x_v(t)=\frac{\sum_{l=1}^{b_v}
y_v^{l}(t)}{\sum_{l=1}^{b_v}
     w_v^l(t)}.
\end{eqnarray*}
\end{enumerate}
\end{algorithm}

\begin{lem}\label{lifted_nonuniform} a) Using
Pseudo-algorithm 2, $\mathbf{x}(t)\rightarrow x_{\mathrm{ave}}\mathbf{1}$.

b) Suppose for the collapsed chain $\mathbf{P}$, there exists some
constant $c'>0$ such that the stationary distribution $\pi_v \geq
\frac{c'}{n}$ for all $v\in V$. Then Algorithm 2 has an averaging
time $T_{\mathrm{ave}}(\epsilon) =
O\left(\log{\epsilon^{-1}}T_{\mathrm{fill}}(\mathbf{\tilde{P}},
c)\right)$ for any constant $0<c<1$.
\end{lem}

\begin{proof} a) Denote the stationary distribution of $\mathbf{\tilde{P}}$ by $\tilde{\mbox{\boldmath$\pi$}}$.
By a similar argument as that of Lemma \ref{lifted-uniform},
$\lim_{t\rightarrow \infty} \mathbf{y}(t) =
nx_{\mathrm{ave}}\tilde{\mbox{\boldmath$\pi$}}$ and
$\lim_{t\rightarrow \infty} \mathbf{w}(t) =
n\tilde{\mbox{\boldmath$\pi$}}$. It follows that $\lim_{t\rightarrow
\infty}\mathbf{x}(t) = x_{\mathrm{ave}}\mathbf{1}$.

b) Let $\mathbf{\tilde{p}}(t)$ be the distribution of
$\mathbf{\tilde{P}}$ at time $t$. For any $\epsilon>0$ and any
constant $0<c<1$, Lemma \ref{fillmix} says that there exists some
time
$\tau=O\left(\log{\epsilon^{-1}}T_{\mathrm{fill}}(\mathbf{\tilde{P}},
c)\right)$, such that for any $t\geq \tau$ and any initial
distribution $\mathbf{\tilde{p}}(0)$, \begin{eqnarray}
\|\mathbf{\tilde{p}}(t)-\mbox{\boldmath$\pi$}\|_1\leq
\frac{\epsilon(1-c)c'}{2}.
\end{eqnarray}
Moreover, for $t\geq T_{\mathrm{fill}}(\mathbf{\tilde{P}}, c)$, we have for $\forall v\in V$,
\begin{eqnarray}
\sum_{\tilde{v}\in f^{-1}(v)} w_{\tilde{v}}(t)\geq
(1-c)\sum_{\tilde{v}\in f^{-1}(v)} \tilde{\pi}_{\tilde{v}}(t)n =
(1-c)\pi_vn\geq (1-c)c'.
\end{eqnarray} Thus, for $\forall t\geq \tau$,
\begin{eqnarray*}
&&\|\mathbf{x}(t)-x_{\mathrm{ave}}\mathbf{1}\|_1=\sum_{v\in
V}|x_v(t)-x_{\mathrm{ave}}|\\&=&\sum_{v\in
V}|\frac{\sum_{\tilde{v}\in
f^{-1}(v)}y_{\tilde{v}}(t)}{\sum_{\tilde{v}\in
f^{-1}(v)}w_{\tilde{v}}(t)}-x_{\mathrm{ave}}|
\\&\leq&\frac{1}{(1-c)c'}\sum_{v\in V}|\sum_{\tilde{v}\in f^{-1}(v)}\left(y_{\tilde{v}}(t)-w_{\tilde{v}}(t)x_{\mathrm{ave}}\right)|\\
&\leq&\frac{1}{(1-c)c'}\sum_{\tilde{v}\in \tilde{V}}|y_{\tilde{v}}(t)-w_{\tilde{v}}(t)x_{\mathrm{ave}}|\\
&\leq&\frac{1}{(1-c)c'}\left[\sum_{\tilde{v}\in\tilde{V}}|y_{\tilde{v}}(t)-n\tilde{\pi}_{\tilde{v}}x_{\mathrm{ave}}|+\sum_{\tilde{v}\in\tilde{V}}|w_{\tilde{v}}(t)-n\tilde{\pi}_{\tilde{v}}|x_{\mathrm{ave}}\right]\\
&\leq&\frac{1}{(1-c)c'}\left[\frac{\epsilon(1-c)c'}{2}nx_{\mathrm{ave}}+\frac{\epsilon(1-c)c'}{2}
nx_{\mathrm{ave}}\right]=\epsilon \|\mathbf{x}(0)\|_1.
\end{eqnarray*}
\end{proof}
Remark: It is clear that $w_{\tilde{v}}$ serves to estimate the
scaling factor $n\tilde{\pi}_{\tilde{v}}$ at each iteration.
Alternatively, a pre-computation phase can be employed where each
node $v$ computes $\sum_{\tilde{v}\in
f^{-1}(v)}\tilde{\pi}_{\tilde{v}}$. Then only the $y$ values need to
be communicated.

In the above, we have proposed two pseudo-algorithms to illustrate the idea of distributed consensus through lifting Markov chains, leaving out the details of constructing fast-mixing Markov chains. In the following two sections, we present one efficient realization for each of these two pseudo-algorithms, on regular networks and geometric random networks, respectively.

\section{LADA Algorithm On Grid}
In this section, we present a LADA algorithm on a $k \times k$ grid.
This algorithm utilizes the direction information (not the absolute
geographic location) of neighbors to construct a fast-mixing Markov
chain, and is a specific example of Pseudo-Algorithm 1 described in
Section III. While existing works typically assumes a torus
structure to avoid edge effects and simplify analysis, we consider
the grid structure which is a more realistic model for planar
networks, and explicitly deal with the edge effects. This algorithm
is then extended to a centralized algorithm for general wireless
network as modeled by a geometric random graph. Our analysis
directly addresses the standard definition of mixing time in
(\ref{mixing}). Besides interest in its own right, results in this
section will also facilitate our analysis in the following sections.

\subsection{Algorithm}
Consider a $k\times k$ grid. For each node $i$, denote its east, north, west and south
neighbor (if exists) respectively by $N_i^0$,$N_i^1$,
$N_i^2$ and $N_i^3$, as shown in Fig. \ref{grid_structure}. Each node $i$
maintains four values indexed according to the four directions counter-clockwise (see Fig. \ref{grid_structure}). The east, north, west and south value of node $i$, denoted respectively by
$y_i^{0}$, $y_i^{1}$, $y_i^{2}$ and $y_i^{3}$, are initialized to
\begin{eqnarray} y_i^{l}(0)=\frac{ x_i(0)}{4}, \quad l=0,\cdots,3.
\end{eqnarray}

\begin{figure} \centering
\includegraphics[width=3.0in, bb=147 362 450
680]{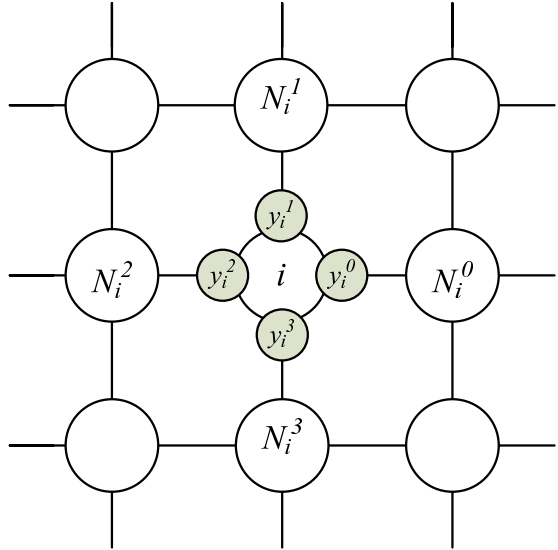}\caption{Node neighbors and values in the grid}\label{grid_structure}
\end{figure}

At each time instant $t$, the east value of node $i$ is updated with
\begin{eqnarray}
y_i^{0}(t+1)=\left(1-\frac{1}{k}\right)y_{N_i^2}^{0}(t)
+\frac{1}{2k}\left(y_{N_i^2}^{1}(t)+y_{N_i^2}^{3}(t)\right).
\end{eqnarray}
That is, the east value of $i$ is updated by a weighted sum of the
previous values of its west neighbor, with the majority
($1-\frac{1}{k}$) coming from the east value, and a fraction of
$\frac{1}{2k}$ coming from the north value as well as the south
value. If $i$ is a west border node (i.e., one without a west
neighbor), then the west, north and south value of itself are used
as substitutes:
\begin{eqnarray}
y_i^{0}(t+1)=\left(1-\frac{1}{k}\right)y_i^{2}(t)
+\frac{1}{2k}\left(y_i^1(t)+y_i^3(t)\right).
\end{eqnarray}
The above discussion is illustrated in Fig. \ref{grid_incoming}.
Intuitively the west value is ``bounced back" when it reaches the
west boundary and becomes the east value. As we will see, this is a
natural procedure on the grid structure to ensure that the iteration
evolves according to a doubly stochastic matrix which is desirable
for averaging. Moreover, the fact that the information continues to
propagate when it reaches the boundary is essential for the
associated chain to mix rapidly. Similarly, the north value of $i$
is updated by a weighted sum of the previous values of its south
neighbor, with the majority coming from the north value, and so on.
Each node then calculates the average of its four values as an
estimate for the global average:
\begin{eqnarray}\label{x2}
x_i(t+1)=\sum_{l=0}^3 y_{i}^l(t+1).
\end{eqnarray}

\begin{figure} \centering
\includegraphics[width=4.0in, bb=38 38 611
275]{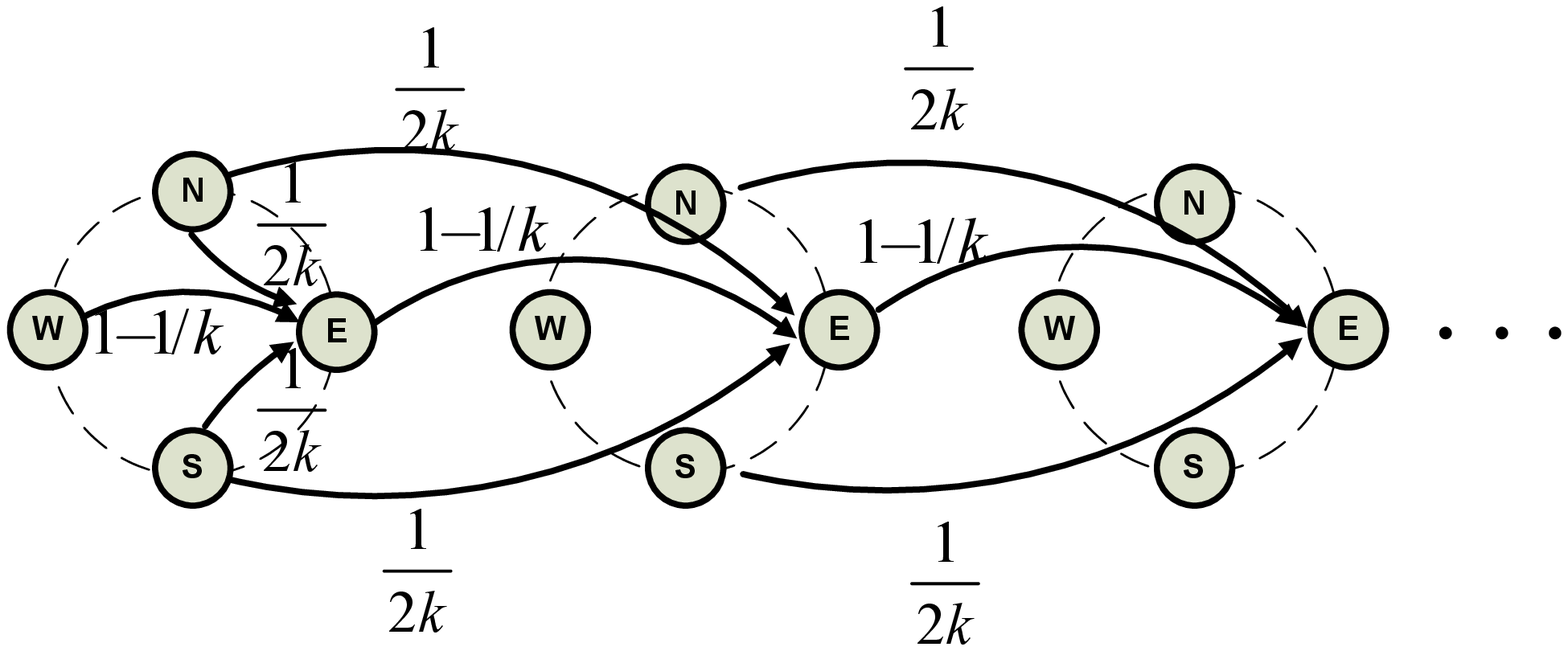}\caption{Updating of east values for a normal
node (right) and a west boundary node (left)}\label{grid_incoming}
\end{figure}

%

\subsection{Analysis}
Assume nodes in the $k\times k$ grid are indexed by $(x, y) \in
[0,k-1]\times[0, k-1]$, starting from the south-west corner. The
nonreversible Markov chain $\tilde{\mathbf{P}}$ underlying the above
algorithm is illustrated in Fig. \ref{grid}. Each state $s\in
\mathcal{S}$ is represented by a triplet $s=(x, y, l)$, with
$l\in\{\mathrm{E,W,N,S}\}$ denoting the specific state within a node
in terms of its direction. The transition probabilities of
$\tilde{\mathbf{P}}$ for an east node are as follows (similarly for
$l\in\{\mathrm{N,W,S}\}$):
\begin{eqnarray}
&&\tilde{\mathbf{P}}\left((x,y,\mathrm{E}),~ (x+1,y,\mathrm{E})\right) = 1-\frac{1}{k},\quad x<k-1\\
&&\tilde{\mathbf{P}}\left((x,y,\mathrm{E}), ~(x,y,\mathrm{W})\right) = 1-\frac{1}{k},\quad x=k-1\\
&&\tilde{\mathbf{P}}\left((x,y,\mathrm{E}),
~(x,y+1,\mathrm{N})\right) =
\tilde{\mathbf{P}}\left((x,y,\mathrm{E}),~ (x,y-1,\mathrm{S})\right)
= \frac{1}{2k}, \quad 0<y<k-1\\
&&\tilde{\mathbf{P}}\left((x,y,\mathrm{E}), ~(x,y,\mathrm{S})\right)
= \tilde{\mathbf{P}}\left((x,y,\mathrm{E}),~
(x,y-1,\mathrm{S})\right)
= \frac{1}{2k}, \quad y=k-1\\
&&\tilde{\mathbf{P}}\left((x,y,\mathrm{E}),
~(x,y+1,\mathrm{N})\right) =
\tilde{\mathbf{P}}\left((x,y,\mathrm{E}),~ (x,y,\mathrm{N})\right) =
\frac{1}{2k}, \quad y=0.
\end{eqnarray}
It can be verified that $\tilde{\mathbf{P}}$ is doubly stochastic,
irreducible and aperiodic. Therefore, $\tilde{\mathbf{P}}$ has a
uniform stationary distribution on its state space, and so does its
collapsed chain. Consequently each $x_i(t)\rightarrow
x_{\mathrm{ave}}$ by Lemma \ref{lifted-uniform}. Moreover, since the
nonreversible random walk $\tilde{\mathbf{P}}$ most likely keeps its
direction, occasionally makes a turn, and never turns back, it mixes
substantially faster than a simple random walk (where the next
node is chosen uniformly from the neighbors of the current node). 
Our main results on the mixing time of this chain, and the averaging
time of the corresponding LADA algorithm are given below.

\begin{figure} \centering
\includegraphics[width=4.0in, bb=80 320 500 740
250]{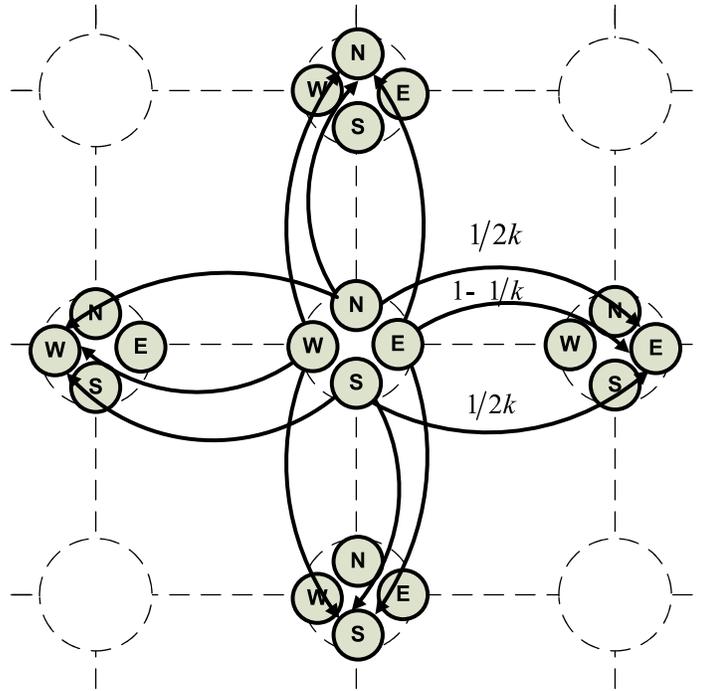}\caption{Nonreversible chain used in the LADA algorithm on a grid: outgoing probabilities for the states of node $i$ are
depicted.}\label{grid}
\end{figure}

\begin{lem}\label{gridmix}
The $\epsilon$-mixing time of $\tilde{\mathbf{P}}$ is  a)
$T_\mathrm{mix}(\tilde{\mathbf{P}},\epsilon) =
O(k\log(\epsilon^{-1}))$, for any $\epsilon
>0$;
\\b) $T_\mathrm{mix}(\tilde{\mathbf{P}}, \epsilon) = \Theta(k)$, for
a sufficiently small constant $\epsilon$.
\end{lem}

\begin{proof}
a) See Appendix \ref{appgrid}. The key is to show that
$T_\mathrm{fill}=O(k)$. The desired result then follows from Lemma
\ref{fillmix}.

b) We are left to show that $T_\mathrm{mix}(\tilde{\mathbf{P}},
\epsilon) = \Omega(k)$ for a constant $\epsilon$ which is
sufficiently small (less than 2/32 in this case). For the random
walk starting from $s_0\in\mathcal{S}$, denote by $\hat{s}_t$ the
state it visits at time $t$ if it never makes a turn. Note that
$\left(1-\frac{1}{k}\right)^{k}$ is an increasing function in $k$,
hence $\left(1-\frac{1}{k}\right)^{k} \geq \frac{1}{4}$ for $k\geq
2$. Thus we have for $t\leq k$,
\begin{eqnarray}
&&\|\tilde{\mathbf{P}}^t(s_0,\cdot)-\frac{1}{4k^2}\cdot\mathbf{1}\|_1
\geq |\tilde{\mathbf{P}}^t(s_0,\hat{s}_t)- \frac{1}{4k^2} |
= |\left(1-\frac{1}{k}\right)^t- \frac{1}{4k^2} |\\
 &\geq&
\left(1-\frac{1}{k}\right)^{k}-\frac{1}{4k^2}\geq\frac{1}{4}-\frac{1}{16}=\frac{3}{16}>2\epsilon,
\end{eqnarray}
for $0<\epsilon<\frac{3}{32}$, where the second inequality follows
from
$\left(1-\frac{1}{k}\right)^{t}\geq\left(1-\frac{1}{k}\right)^{k}\geq
\frac{1}{4}\geq \frac{1}{4k^2}$. The result follows from the
definition of mixing time in (\ref{mixing}).
\end{proof}

\begin{thm}\label{gridave} For the LADA algorithm on a $k\times k$ grid,
a) $T_{\mathrm{ave}}(\epsilon)=O(k\log(\epsilon^{-1}))$ for any $\epsilon>0$; \\
b) $T_{\mathrm{ave}}(\epsilon)=\Theta(k)$ for a sufficiently small constant $\epsilon$. 
\end{thm}

\begin{proof} a) Follows from Lemma \ref{lifted-uniform} and Lemma \ref{gridmix} a).

b) Note that the proof of Lemma \ref{gridmix} b) also implies that
for $k\geq 3$, for any initial state $s_0\in \mathcal{S}$, when
$t\leq k$, there is at least one state $\hat{s}\in \mathcal{S}$ with
which
$\tilde{\mathbf{P}}^t(s_0,\hat{s})\geq\left(1-\frac{1}{k}\right)^{k}\geq
\frac{8}{27}$. Suppose state $\hat{s}$ is some state belonging to
some node $v$. Thus for $t\leq k$ ($k\geq 3$)
\begin{eqnarray}\label{thmgrid}
|x_v(t)-x_{\mathrm{ave}}|=|\sum_{s\in
f^{-1}(v)}\tilde{\mathbf{P}}^t(s_0,s)-
\frac{1}{k^2}|\cdot\|\mathbf{x}(0)\|_1\geq|\tilde{\mathbf{P}}^t(s_0,\hat{s})-
\frac{1}{k^2}|\cdot\|\mathbf{x}(0)\|_1\geq\frac{5}{27}\|\mathbf{x}(0)\|_1,
\end{eqnarray}
i.e, node $v$ has not reached an average estimate in this scenario
(when $0<\epsilon<\frac{5}{27}$).
\end{proof}

\subsection{A Centralized Grid-based Algorithm for Wireless Networks}
The regular grid structure considered above does appear in some
applications, and often serves as a first step towards modeling a
realistic network. In this section, we explore a celebrated model
for wireless networks, geometric random graphs, and present a
centralized algorithm which achieves an $\epsilon$-averaging time of
$O(r^{-1}\log(\epsilon^{-1}))$ on $G(n,r)$. The algorithm relies on
a central controller to perform tessellation and clustering, and
simulates the LADA algorithm on the grid proposed above on the
resultant 2-d grid. This is a common approach in literature (e.g.,
\cite{Gupta}), where the main purpose is to explore the best
achievable performance in wireless networks, with implementation
details ignored.

Assume that the unit area is tesselated into $k^2 \triangleq
\ulcorner{\frac{\sqrt{5}}{r}}\urcorner ^2$ squares (clusters). By
this tessellation, a node in a given cluster is adjacent to all
nodes in the four edge-neighboring clusters. Denote the number of
nodes in a given cluster $m$ by $n_m$. Then for a geometric random
graph $n_m\geq 1$ for all $m$ w.h.p.\cite{Gupta}. One node in each
cluster is selected as a cluster-head. Denote the index of the
cluster where node $i$ lies by $C_i$. For each cluster $m$, denote
its east, north, west and south neighboring cluster (if exists)
respectively by $N_m^0$,$N_m^1$, $N_m^2$ and $N_m^3$. Every
cluster-head maintains four values corresponding to the four
directions from east to south clockwise, denoted respectively by
$y_m^{0}$, $y_m^{1}$, $y_m^{2}$ and $y_m^{3}$ for cluster $m$. In
the initialization stage, every node transmits its value to the
cluster-head. The cluster-head of cluster $m$ computes the sum of
the values within the cluster and initializes all its four values to
\begin{eqnarray} y_m^{l}(0)=\frac{1}{4}\sum_{C_i=m} x_i(0), \quad l=0,\cdots,3.
\end{eqnarray}
At each time instant $t$, the cluster-heads of neighboring clusters
communicate and update their values following exactly the same rules
as the LADA algorithm on the grid. Each cluster-head then calculates
the average of its four values as an estimate for the global
average, and broadcasts this estimate to its members, so that every
node $i$ obtains
\begin{eqnarray}\label{x2}
x_i(t+1)=\frac{k^2}{n}\sum_{l=0}^3 y_{C_i}^l(t+1).
\end{eqnarray}

\begin{thm}\label{centralave} The centralized algorithm has an $\epsilon$-averaging time
$T_{\mathrm{ave}}(\epsilon)=O(r^{-1}\log(\epsilon^{-1}))$ on the
geometric random graph $G(n,r)$ with common transmission radius
$r>\sqrt{\frac{20\log n}{n}}$ w.h.p. Moreover, for a sufficiently
small constant $\epsilon$,
$T_{\mathrm{ave}}(\epsilon)=\Theta(r^{-1})$.
\end{thm}

\begin{proof} We can appeal to uniform convergence in the law of large
numbers using Vapnik-Chervonenkis theory as in \cite{Gupta} to bound
the number of nodes in each cluster as follows:
\begin{eqnarray}
\Pr \left(\max_{1\leq m\leq k^2} |
\frac{n_m}{n}-\frac{1}{k^2}|\leq\epsilon(n)\right)>1-\delta(n)
\end{eqnarray}
when $n\geq \max\{\frac{3}{\epsilon(n)}\log\frac{16e}{\epsilon(n)},
\frac{4}{\epsilon(n)}\log \frac{2}{\delta(n)}\}$. This is satisfied
if we choose $ \epsilon(n) = \delta(n) = \frac{4\log n}{n}$. Thus we
have for all $m$, $n_m\geq \frac{n}{k^2}-4\log
n=\frac{nr^2}{5}-4\log n$, which is at least 1 for sufficiently
large $n$ if $r>\sqrt{\frac{20\log n}{n}}$. In this case, we have
that $\frac{c_2n}{k^2}\leq n_m\leq \frac{c_1n}{k^2}$ for all $m$ for
some constants $c_1, c_2>0$ w.h.p. By Lemma \ref{gridmix} a), for
any $\epsilon> 0$, there exists some $\tau =
T_{\mathrm{mix}}(\tilde{\mathbf{P}}, \frac{\epsilon}{2c_1}) =
O(r^{-1}\log(\epsilon^{-1}))$ such that for all $t\geq \tau$,
\begin{eqnarray*}
&&\|\mathbf{x}(t)-x_{\mathrm{ave}}\mathbf{1}\|_1 = \sum_{m=1}^{k^2}
n_m |\frac{k^2}{n}\sum_{l=0}^3y_m^l(t)-x_{\mathrm{ave}}|\leq
\sum_{m=1}^{k^2}
\frac{n_mk^2}{n}\sum_{l=0}^3|y_m^l(t)-\frac{nx_{\mathrm{ave}}}{4k^2}|\\&\leq&
\epsilon\|\mathbf{x}(0)\|_1,
\end{eqnarray*}
where the last inequality follows a similar argument as in the proof
of Lemma \ref{lifted-uniform}.

To prove the latter part of the theorem, note that
$\|\mathbf{x}(t)-x_{\mathrm{ave}}\mathbf{1}\|_1\geq
c_2\sum_{m=1}^{k^2}|\sum_{l=0}^3y_m^l(t)-\frac{nx_{\mathrm{ave}}}{k^2}|$.
The rest follows a similar argument as in the proof of Theorem
\ref{gridave} b).
\end{proof}

In large dynamic wireless networks, it is often impossible to have a
central controller that maintains a global coordinate system and
clusters the nodes accordingly. In the following sections, we
investigate some more practical algorithms, which can be applied to
wireless networks with no central controller or global knowledge
available to nodes.

\section{Distributed LADA Algorithm for Wireless Networks}
In practice, purely distributed algorithms requiring no central
coordination are typically preferred. In this section, we propose a
fully distributed LADA algorithm for wireless networks, which is an
instantiation of Pseudo-Algorithm 2 in Section III. As we mentioned,
while our analysis is conducted on $G(n, r(n))$, our design can
generally be applied to any network topology.

\subsection{Neighbor Classification}
As the LADA algorithm on a grid, LADA for general wireless networks utilizes coarse location information of neighbors to construct fast-mixing nonreversible chains. Due to irregularity of node locations, a neighbor classification procedure is needed.
Specifically, a neighbor $j$ of node $i$ is said to be a Type-$l$
neighbor of $i$, denoted as $j\in \mathcal{N}_i^{l}$, if
\begin{equation}
\angle(X_j-X_i)\in \left(\frac{l\pi}{2}-\frac{\pi}{4},
\frac{l\pi}{2}+\frac{\pi}{4}\right]           \quad l=0,\cdots,3,
\end{equation}
where $X_i$ denotes the geometric location
of node $i$ (whose accurate information is not required). That is, each neighbor $j$ of $i$ belongs to one of the
four regions each spanning 90 degrees, corresponding to east (0),
north (1), west (2) and south (3). Note that if $i\in \mathcal{N}_j^{l}$, then $j\in \mathcal{N}_i^{l+2~\mathrm{mod} 4}$.
We denote the number of type $l$ neighbors for node $i$ by
$d_i^l\triangleq|\mathcal{N}_i^l|$ (except for boundary cases
discussed below).

In literature, wireless networks are often modeled on a unit torus or sphere to avoid the edge effects in
performance analysis \cite{Gupta,
Boyd_TIT}. In our study, we explicitly deal with the edge effects by
considering the following modification, as illustrated in Fig.
\ref{edge}. A boundary node is a node within distance $r$ from one
of the boundaries, e.g., node $i$ in Fig. \ref{edge}. For a boundary
node $i$, we create mirror images of its neighbors with respect to
the boundary. If a neighbor $j$ has an image located within the
transmission range of $i$, node $j$ (besides its original role) is
considered as a virtual neighbor of $i$, whose direction is
determined by the image's location with respect to the location of
$i$. For example, in Fig. \ref{edge}, node $j$ is both a north and a
virtual east neighbor of $i$, and node $i$ is a virtual east
neighbor of itself. Specifically, we use $\widetilde{\mathcal{N}}_i^0$ to denote the set of virtual east neighbors of an
east boundary node $i$, and use $\widehat{\mathcal{N}}_i^0$ to denote the set of virtual east
neighbors of a north or south boundary node $i$. Similarly,
$\widetilde{\mathcal{N}}_i^1$ denotes the set of virtual north
neighbors of a north boundary node $i$, and $\widehat{\mathcal{N}}_i^1$ denotes that of an east or west boundary node, and so on for virtual west and south
neighbors. Informally, $~~\widetilde{}~~$ is used for the case the direction of the virtual neighbors and the boundary ``match", while $~~\widehat{}~~$ is used for the ``mismatch" scenarios.
As we will see, they play different roles in the LADA algorithm. For example, in
Fig. \ref{edge}, we have $i,j,k\in \widetilde{\mathcal{N}}_i^0$, and
$l \in\widehat{\mathcal{N}}_i^3$. It can be shown that if $i\in
\widetilde{\mathcal{N}}_j^l$, then $j\in
\widetilde{\mathcal{N}}_i^l$, while if $i\in
\widehat{\mathcal{N}}_j^l$, then $j\in
\widehat{\mathcal{N}}_i^{l+2~(\mathrm{mod}~4)}$. For a boundary node
$i$, $d_i^l$ is instead defined as the total number of physical and virtual
neighbors in direction $l$, i.e.,
$d_i^l\triangleq|\mathcal{N}_i^l|+|\mathcal{\widetilde{N}}_i^l|+|\mathcal{\widehat{N}}_i^l|$.
With this modification, every type-$l$ neighborhood has an effective
area $\frac{\pi r^2}{4}$, hence $d_i^l$ is roughly the same for all
$i$ and $l$. We also expect that as $n$ increases, the fluctuation
in $d_i^l$ diminishes. This is summarized in the following lemma,
which will be used in our subsequent analysis.

\begin{figure} \centering
\includegraphics[width=3.0in, bb=100 250 440
570]{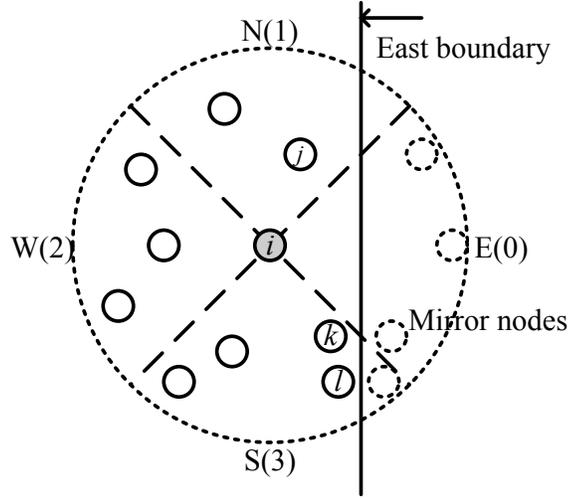}\caption{Illustration of neighbor classification and
virtual neighbors for boundary nodes. Note that for an east boundary
node $i$, there can only be virtual east neighbors of the first
category ($i,j,k\in \widetilde{\mathcal{N}}_i^0$), and virtual north
and south neighbors of the second category ($l
\in\widehat{\mathcal{N}}_i^3$)}\label{edge}
\end{figure}

\begin{lem}\label{regularity} With high probability, the number of type $l$ neighbors of $i$
satisfies\footnote{The stronger result regarding
$r=\Omega\left(\left(\frac{\log n}{n}\right)^{\frac{1}{3}}\right)$
is required for the LADA-U algorithm presented in Appendix C.}
\begin{eqnarray}
d_i^l=\left\{\begin{array}{cc}
          \Theta(nr^2) & \quad \mathrm{if} \quad r>\sqrt{\frac{16\log n}{\pi n}} \\
          \frac{n\pi r^2}{4}\left(1\pm O\left(r\right)\right) &
          \quad
\mathrm{if} \quad r=\Omega\left(\left(\frac{\log
n}{n}\right)^{\frac{1}{3}}\right).
\end{array}\right.
\end{eqnarray}
\end{lem}

\begin{proof}
We can appeal to the Vapnik-Chervonenkis theory as in \cite{Gupta} to bound
the number of nodes in each cluster as follows:
\begin{eqnarray}
\Pr\{\sup_{i,l}|\frac{d_i^l}{n}-\frac{\pi r^2}{4}|\leq\frac{4\log
n}{n}\}>1-\frac{4\log n}{n}.
\end{eqnarray}
Hence, we have $|d_i^l-\frac{n\pi r^2}{4}|\leq 4\log n$ with
probability at least $1-\frac{4\log n}{n}$ for all node $i$ and
direction $l$. Therefore, if $r>\sqrt{\frac{16\log n}{\pi n}}$, we
have $d_i^l=\frac{n\pi r^2}{4}\left(1\pm
O\left(\frac{\log n}{nr^2}\right)\right)=\Theta(nr^2)$. 
If $r=\Omega\left(\left(\frac{\log
n}{n}\right)^{\frac{1}{3}}\right)$, we have $d_i^l=\frac{n\pi
r^2}{4}\left(1\pm O\left(\left(\frac{\log
n}{n}\right)^{\frac{1}{3}}\right)\right)= \frac{n\pi
r^2}{4}\left(1\pm O\left(r\right)\right)$.
\end{proof}

\subsection{Algorithm}
The LADA algorithm for general wireless networks works as follows.
Each node $i$ holds four pairs of values $(y_i^l, w_i^l)$,
$l=0,\cdots,3$ corresponding to the four directions
counter-clockwise: east, north, west and south. The values are
initialized with
\begin{eqnarray}\label{init3}
y_i^{l}(0)=\frac{x_i(0)}{4}, \quad w_i^{l}(0)=\frac{1}{4}, \quad l=0,\cdots,3.
\end{eqnarray}
At time $t$, each node $i$ broadcasts its four values. In turn, it
updates its east value $y_i^{0}$ with
\begin{eqnarray}\label{LADA-I}
y_i^{0}(t+1)=\sum_{j\in
\mathcal{N}_i^2}\frac{1}{d_j^0}\left[(1-p)y_{j}^{0}(t)+\frac{p}{2}\left(y_{j}^{1}(t)+y_{j}^{3}(t)\right)\right],
\end{eqnarray}
where $p=\Theta(r)$ is assumed. This is illustrated in Fig. \ref{LADA_alg1}. That is, the east value of node $i$ is
updated by a sum contributed by all its west neighbors $j\in
\mathcal{N}_i^2$; each contribution is a weighted sum of the
values of node $j$ in the last slot, with the major portion
$\frac{1-p}{d_j^0}$ coming from the east value, and a fraction of
$\frac{p}{2d_j^0}$ coming from the north as well as the south value.
\begin{figure} \centering
\includegraphics[width=4.0in, bb=0 105 400
440]{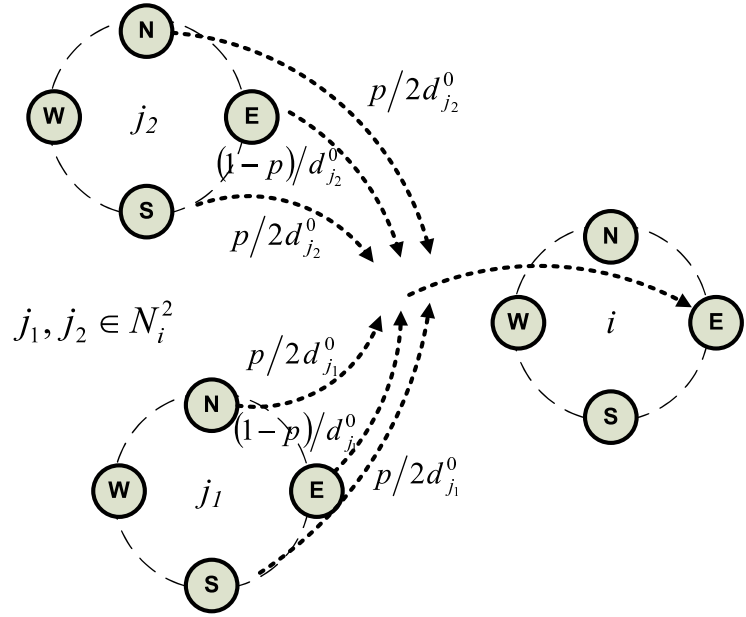}\caption{Update of east value of a normal node $i$:
weighted sums of the east, north and south values of west neighbors
$j_1$, $j_2$}\label{LADA_alg1}
\end{figure}

As in the grid case, boundary
nodes must be treated specially. Let us consider two specific cases:
\begin{enumerate}
\item
If $i$ is a west boundary node (as shown in Fig. \ref{LADA_alg2}),
then we must include an additional term
\begin{eqnarray}
\sum_{j\in
\widetilde{\mathcal{N}}_i^2}\frac{1}{d_j^2}\left[(1-p)y_{j}^{2}(t)+\frac{p}{2}\left(y_{j}^{1}(t)+y_{j}^{3}(t)\right)\right]
\end{eqnarray}
in (\ref{LADA-I}), i.e. values from both physical and virtual west
neighbors (of the first category) are used. Moreover, for the
virtual west neighbors, the west rather than east values are used.
This is similar to the grid case, where the west values are bounced
back and become east values when they reach the west boundary, so
that the information continues to propagate. The factor
$\frac{1}{d_j^2}$ rather than $\frac{1}{d_j^0}$ is adopted here to
ensure the outgoing probabilities of each state of each node $j\in
\widetilde{\mathcal{N}}_i^2$ sum to 1.
\item If $i$ is a north or south
boundary node (as shown in Fig. \ref{LADA_alg3}), however, the sum in (\ref{LADA-I}) is replaced with
\begin{eqnarray}
\sum_{j\in
\mathcal{N}_i^2\bigcup\widehat{\mathcal{N}}_i^2}\frac{1}{d_j^0}\left[(1-p)y_{j}^{0}(t)+\frac{p}{2}\left(y_{j}^{1}(t)+y_{j}^{3}(t)\right)\right],
\end{eqnarray}
i.e., the east, north and south values of both physical and virtual
west neighbors (of the second category) are used. Note that
$\widehat{\mathcal{N}}_i^2$ are meant only for compensating the loss
of neighbors for north or south boundary nodes, so unlike the
previous case, their east or west values continue to propagate in
the usual direction.
\end{enumerate}
If $i$ is both a west and north (or south) boundary node, the above
two cases should be combined. The purpose of introducing virtual
neighbors described above is to ensure the approximate regularity of
the underlying graph of the associated chain, so that the randomized
effect is evenly spread out over the network. The north, west and
south values, as well as the corresponding $w$ values are updated in
the same fashion. Node $i$ computes its estimate of
$x_{\mathrm{ave}}$ with
\begin{eqnarray}\label{x3}
x_i(t+1)=\frac{\sum_{l=0}^3y_i^l(t+1)}{\sum_{l=0}^3w_i^l(t+1)}.
\end{eqnarray}
The detailed algorithm is given in Algorithm 3\footnote{We do not
explicitly differentiate between the non-boundary and boundary
cases, since the corresponding terms are automatically zero for
non-boundary nodes.}.

\begin{figure} \centering
\includegraphics[width=3.0in, bb=158 144 440
500]{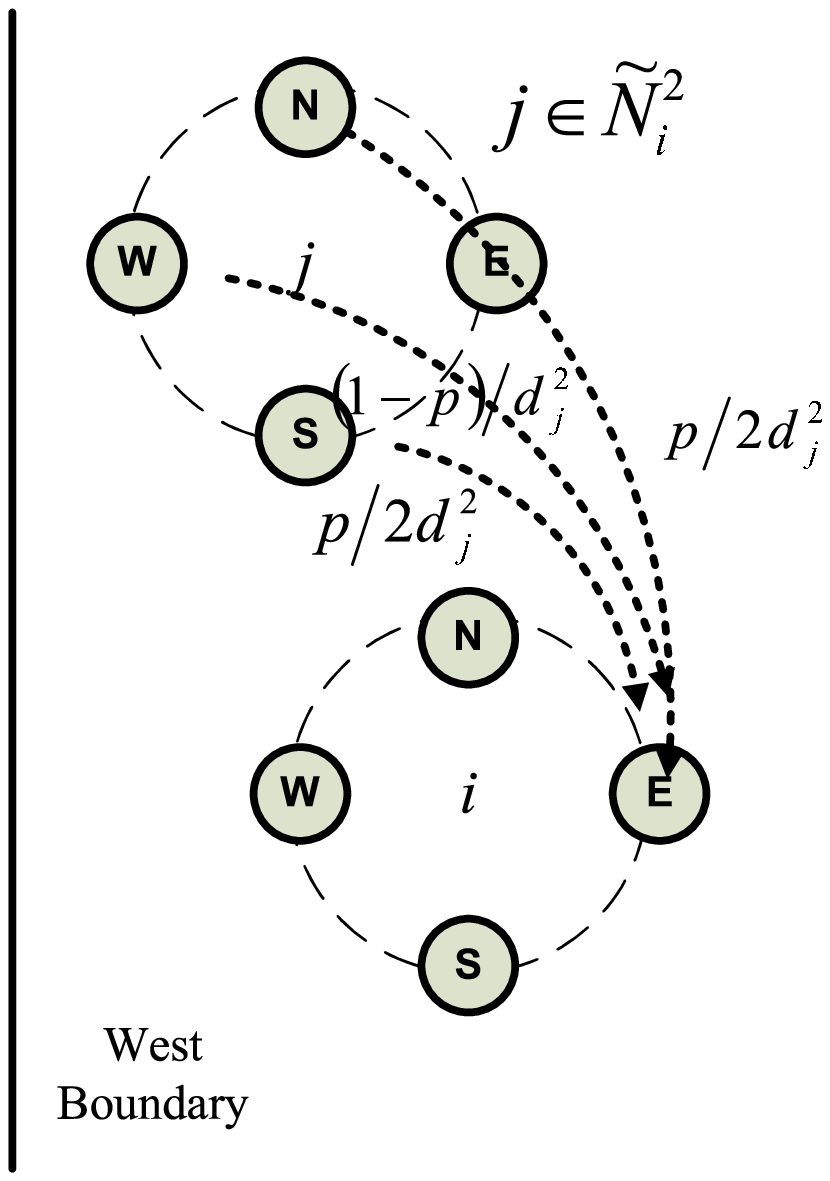}\caption{Update of east value of a west boundary node $i$: west value of virtual west neighbor $j\in\mathcal{\tilde{N}}_i^2$ is used}\label{LADA_alg2}
\end{figure}

\begin{figure} \centering
\includegraphics[width=4.0in, bb=55 180 460
390]{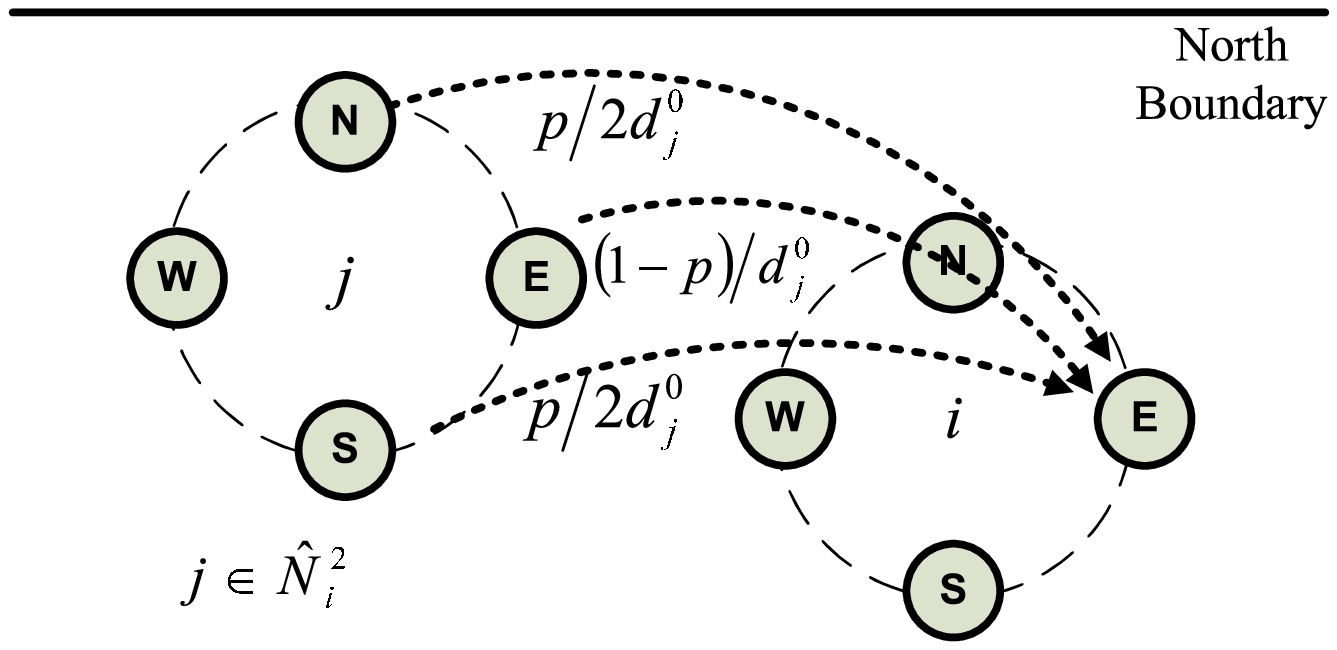}\caption{Update of east value of a north boundary node $i$: east value of virtual west neighbor $j\in\mathcal{\hat{N}}_i^2$ is used} \label{LADA_alg3}
\end{figure}

\begin{algorithm}
\caption{LADA Algorithm}
\begin{algorithmic}
\FOR{$i=1$ to $n$} \STATE $y_i^{l}(0)\Leftarrow x_i(0)$,
$w_i^{l}(0)\Leftarrow 1$, $l=0,1,2,3$ \ENDFOR \STATE $p\Leftarrow
\frac{r}{2}$, $t\Leftarrow 0$
\WHILE{$\|\mathbf{x}(t)-x_{\mathrm{ave}}\mathbf{1}\|_1> \epsilon$}
\FOR{$i=1$ to $n$} \FOR{$l=0$ to 3} \STATE $y_i^{l}(t+1)\Leftarrow
\sum_{j\in \mathcal{N}_i^{\overline{l+2}}\bigcup
\widehat{\mathcal{N}}_i^{\overline{l+2}}}\frac{1}{d_j^l}\left[(1-p)y_{j}^{l}(t)+\frac{p}{2}\left(y_{j}^{\overline{l+1}}(t)+y_{j}^{\overline{l+3}}(t)\right)\right]
+\sum_{j\in
\widetilde{\mathcal{N}}_i^{\overline{l+2}}}\frac{1}{d_j^{\overline{l+2}}}\left[(1-p)y_{j}^{\overline{l+2}}(t)+\frac{p}{2}\left(y_{j}^{\overline{l+1}}(t)+y_{j}^{\overline{l+3}}(t)\right)\right]$
\STATE $w_i^{l}(t+1)\Leftarrow \sum_{j\in
\mathcal{N}_i^{\overline{l+2}}\bigcup
\widehat{\mathcal{N}}_i^{\overline{l+2}}}\frac{1}{d_j^l}\left[(1-p)w_{j}^{l}(t)+\frac{p}{2}\left(w_{j}^{\overline{l+1}}(t)+w_{j}^{\overline{l+3}}(t)\right)\right]
+\sum_{j\in
\widetilde{\mathcal{N}}_i^{\overline{l+2}}}\frac{1}{d_j^{\overline{l+2}}}\left[(1-p)w_{j}^{\overline{l+2}}(t)+\frac{p}{2}\left(w_{j}^{\overline{l+1}}(t)+w_{j}^{\overline{l+3}}(t)\right)\right]$
\ENDFOR  \STATE
$x_i(t+1)\Leftarrow\frac{\sum_{l=0}^3y_i^l(t+1)}{\sum_{l=0}^3w_i^l(t+1)}$
\ENDFOR \STATE $t\Leftarrow t+1$ \ENDWHILE
\end{algorithmic}
\end{algorithm}

We remark that even the exact knowledge of directions is not
critical for the LADA algorithm. For example, if a neighbor $j$ of
node $i$ is roughly on the border of two regions, it is fine to
categorize $j$ to either region, as long as $j$ categorizes $i$
correspondingly (i.e., $i\in \mathcal{N}_j^{l+2~(\mathrm{mod} 4)}$
if $j\in \mathcal{N}_i^l$).

\subsection{Analysis}
Denote $\mathbf{y}=[\mathbf{y}_{0}^T, \mathbf{y}_{1}^T
\mathbf{y}_{2}^T,\mathbf{y}_{3}^T]^T$, with $\mathbf{y}_l=[y_1^l,
y_2^l, \cdots, y_n^l]^T$, and similarly denote $\mathbf{w}$. The
above iteration can be written as
$\mathbf{y}(t+1)=\tilde{\mathbf{P}}_1^T\mathbf{y}(t)$ and
$\mathbf{w}(t+1)=\tilde{\mathbf{P}}_1^T\mathbf{w}(t)$. Using the
fact that if $i\in \mathcal{N}_j^l\bigcup\widehat{\mathcal{N}}_j^l$,
then $j\in
\mathcal{N}_i^{l+2~(\mathrm{mod}~4)}\bigcup\widehat{\mathcal{N}}_i^{l+2~(\mathrm{mod}~4)}$,
and if $i\in \widetilde{\mathcal{N}}_j^l$, then $j\in
\widetilde{\mathcal{N}}_i^l$, it can be shown that each row in
$\tilde{\mathbf{P}}_1$ (i.e., each column in
$\tilde{\mathbf{P}}_1^T$) sums to 1, hence $\tilde{\mathbf{P}}_1$ is
a stochastic
matrix (see Fig. \ref{LADA} for an illustration). 
On a finite connected 2-d network, the formed chain
$\tilde{\mathbf{P}}_1$ is irreducible and aperiodic by construction.
Since the incoming probabilities of a state do not sum to 1 (see Eq.
(\ref{LADA-I}) and Fig. \ref{LADA_alg1})\footnote{Due to
irregularity of the network, all west neighbors of a node don't have
exactly the same number of east neighbors.}, $\tilde{\mathbf{P}}_1$
is not doubly stochastic and does not have a uniform stationary
distribution. The LADA algorithm for general wireless networks is a
special case of the Pseudo-Algorithm 2 in Section III, and it
converges to the average of node values by Lemma
\ref{lifted_nonuniform} a). In the rest of this section, we analyze
the performance of LADA algorithm on geometric random graphs.

\begin{figure} \centering
\includegraphics[width=4.0in, bb=60 100 470
420]{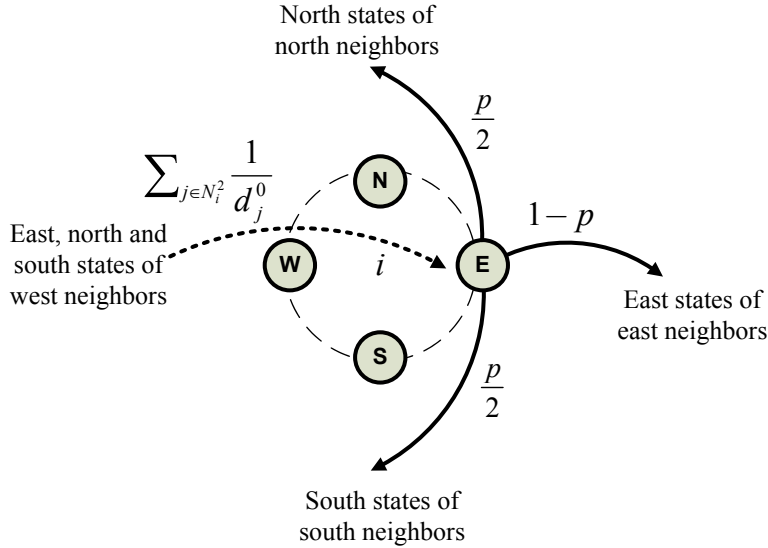}\caption{The Markov chain used in LADA: combined outgoing
probabilities (solid lines) and combined incoming probabilities
(dotted line) for the east state of node $i$ are
depicted}\label{LADA}
\end{figure}

%


\begin{lem}\label{LADAmix} On the geometric random graph $G(n,r)$ with $r=\Omega\left(\sqrt{\frac{\log n}{n}}\right)$,
with high probability, the Markov chain $\tilde{\mathbf{P}}_1$
constructed in the LADA algorithm has an approximately uniform
stationary distribution, i.e., for any $s\in \mathcal{S}$,
$\pi(s)=\Theta\left(\frac{1}{4n}\right)$, and
$T_{\mathrm{fill}}(\tilde{\mathbf{P}}_1, c)=O(r^{-1})$ for some
constant $0<c<1$.
\end{lem}

The proof is given in Appendix \ref{appLADA}. Essentially, we first
consider the expected location of the random walk
$\tilde{\mathbf{P}}_1$ (with respect to the node distribution),
which is shown to evolve according to the random walk
$\tilde{\mathbf{P}}$ on a $k\times k$ grid with $k=\Theta(r^{-1})$
when $p=\Theta(r)$. Thus the expected location of
$\tilde{\mathbf{P}}_1$ can be anywhere on the grid in $O(k)$ steps
(see Section IV). Then, we take the random node location into
account and further show that when $n\rightarrow \infty$, the exact
location of the random walk $\tilde{\mathbf{P}}_1$ can be anywhere in the network in $O(r^{-1})$ steps. 

\begin{thm}\label{LADAave} On the geometric random graph $G(n,r)$
with $r=\Omega\left(\sqrt{\frac{\log n}{n}}\right)$,
the LADA algorithm has an $\epsilon$-averaging time
$T_{\mathrm{ave}}(\epsilon)=O(r^{-1}\log(\epsilon^{-1}))$ with high probability.
\end{thm}

\begin{proof}
Since when $r=\Omega\left(\sqrt{\frac{\log n}{n}}\right)$, the
Markov chain $\tilde{\mathbf{P}}_1$ constructed in the LADA
algorithm has an approximately uniform stationary distribution from
Lemma \ref{LADAmix}, so does its collapsed chain. Thus Lemma
\ref{lifted_nonuniform} b) can be invoked to show that
$T_{\mathrm{ave}}(\epsilon)=O\left(T_{\mathrm{fill}}(\tilde{\mathbf{P}}_1,c)\log(\epsilon^{-1})\right)=O(r^{-1}\log(\epsilon^{-1}))$.
\end{proof}


We have also explored a variant of the LADA algorithm, called LADA-U
, which is a realization of Pseudo-Algorithm 1. The nonreversible
chain is carefully designed to ensure a uniform stationary
distribution (accounting for the suffix ``U"), by allowing
transitions between the east and the west, as well as between the
north and south state for each node. It can be shown that LADA-U can
achieve the same scaling law in averaging time as LADA, but
requiring a transmission range larger than the minimum connectivity
requirement, mainly due to the induced diffusive behavior. In
particular, a sufficient condition for the same scaling law as LADA
to hold is $r=\Omega\left(\left(\frac{\log
n}{n}\right)^{\frac{1}{3}}\right)$. The LADA-U algorithm and its
performance analysis are summarized in Appendix \ref{appLADAU} for
possible interest of the reader.

\subsection{$T_{fill}$ Optimality of LADA Algorithm}
To conclude this section, we would like to discuss the following
question: what is the optimal performance of distributed consensus
through lifting Markov chains on a geometric random graph, and how
close LADA performs to the optimum? A straightforward lower bound of
the averaging time of this class of algorithms would be given by the
diameter of the graph, hence
$T_{\mathrm{ave}}(\epsilon)=\Omega(r^{-1})$. Therefore, for a
constant $\epsilon$, LADA algorithm is optimal in the
$\epsilon$-averaging time. For $\epsilon = O(1/n)$, it is not known
whether the lower bound $\Omega(r^{-1})$ can be further tightened,
and whether LADA achieves the optimal $\epsilon$-averaging time in
scaling law. Nevertheless, we provide a  partial answer to the
question by showing that the constructed chain attains the optimal
scaling law of $T_{\mathrm{fill}}(\tilde{\mathbf{P}},c)$ for a
constant $c\in (0,1)$, among all chains lifted from one with an
approximately uniform stationary distribution on $G(n,r)$.
%
For our analysis,  we first introduce two invariants of a Markov
chain, the conductance and the resistance. The conductance measures
the chance of a random walk leaving a set after a single step, and
is defined for the corresponding chain $\mathbf{P}$ as
\cite{Sinclair}
\begin{eqnarray}\label{conductance}
\Phi(\mathbf{P}) = \min_{S\subset V, 0<\pi(S)< 1}
\frac{Q(S,\bar{S})}{\pi(S)\pi(\bar{S})}
\end{eqnarray}
where $\bar{S}$ is the complement of $S$ in $V$,
$Q(A,B)=\sum_{i\in A}\sum_{j\in B}Q_{ij}$, and for edge $e = ij$, $Q(e)=Q_{ij}=\pi_iP_{ij}$ is often
interpreted as the capacity of the edge in combinatorial research.
The resistance is defined in terms of multi-commodity flows. 
A flow\footnote{An alternative and equivalent definition of a flow
as a function on the edges of graphs can be found in
\cite{Bollobas}.} in the underlying graph $G(\mathbf{P})$ of
$\mathbf{P}$ is a function $f:~\Gamma\rightarrow \mathbb{R}^+$ which
satisfies
\begin{eqnarray}
\sum_{\gamma\in \Gamma_{uv}}f(\gamma)=\pi(u)\pi(v)   \quad \forall
u,v\in V, u\neq v
\end{eqnarray}
where $\Gamma_{uv}$ is the set of all simple directed paths from $u$
to $v$ in  $G(\mathbf{P})$ and $\Gamma=\bigcup_{u\neq
v}\Gamma_{uv}$. The congestion parameter $R(f)$ of a flow $f$ is
defined as
\begin{eqnarray}
R(f)\triangleq \max_{e}\frac{1}{Q(e)}\sum_{\gamma\in \Gamma;
\gamma\ni e}f(\gamma).
\end{eqnarray}
The resistance
of the chain $\mathbf{P}$ is defined as the minimum value of $R(f)$
over all flows,
\begin{eqnarray}
R(\mathbf{P})=\inf_fR(f).
\end{eqnarray}

It has been shown that the resistance of an ergodic reversible
Markov chain $\mathbf{P}$ satisfies $R(\mathbf{P})\leq
16T_{\mathrm{mix}}(\mathbf{P},1/8)$\cite{Sinclair}. This result does
not readily apply to nonreversible chains. Instead, a similar result
exists for $T_{\mathrm{fill}}$, as given below.

\begin{lem}\label{resistancefill} For any irreducible and aperiodic Markov chain $\mathbf{P}$,
the resistance satisfies
\begin{eqnarray}
T_{\mathrm{fill}}(\mathbf{P},c)\geq \frac{R(\mathbf{P})}{1-c}.
\end{eqnarray}
\end{lem}

\begin{proof}
Let $t=T_{\mathrm{fill}}(\mathbf{P},c)$. Let $\Gamma^{(t)}_{uv}$
denote the set of all (not necessarily simple) paths of length
exactly $t$ from $u$ to $v$ in the underlying graph $G(\mathbf{P})$.
$\Gamma^{(t)}_{uv}$ is nonempty by the definition of
$T_{\mathrm{fill}}$. For each $\gamma\in \Gamma^{(t)}_{uv}$, let
$p(\gamma)$ denote the probability that the Markov chain, starting
in state $u$, makes the sequence of transitions defined in $\gamma$,
thus $\sum_{\gamma\in \Gamma^{(t)}_{uv}}p(\gamma)=P^t(u,v)$. For
each $u,v$ and $\gamma\in \Gamma^{(t)}_{uv}$, set
\begin{eqnarray}
f(\gamma) = \frac{\pi(u)\pi(v)p(\gamma)}{P^t(u,v)}
\end{eqnarray}
and set $f(\gamma)=0$ for all other paths. Thus, $\sum_{\gamma\in
\Gamma^{(t)}_{uv}}f(\gamma)=\pi(u)\pi(v)$. Now, by removing cycles
on all paths, we can obtain a flow $f'$ (consisting of simple paths)
from $f$ without increasing the throughput on any edge. The flow
routed by $f'$ through $e$ is
\begin{eqnarray}\label{flow-bound}
f'(e)\triangleq\sum_{\gamma\in \Gamma; \gamma\ni e}f'(\gamma)\leq
\sum_{u,v}\sum_{\gamma\in \Gamma^{(t)}_{uv}, \gamma \ni e}
\frac{\pi(u)\pi(v)p(\gamma)}{P^t(u,v)}\leq \frac{1}{1-c}
\sum_{u,v}\sum_{\gamma\in \Gamma^{(t)}_{uv}, \gamma \ni e}
\pi(u)p(\gamma),
\end{eqnarray}
where the second inequality follows from the definition of
$T_{\mathrm{fill}}$. The final double sum in (\ref{flow-bound}) is
precisely the probability that the stationary process traverses the
oriented edge $e$ within $t$ steps, which is at most $tQ(e)$. It
then follows
\begin{eqnarray}
R(f')=\max_{e}\frac{f'(e)}{Q(e)}\leq \frac{t}{1-c}.
\end{eqnarray}
\end{proof}

%

\begin{lem}\label{lowerresistance} For the geometric random graph $G(n, r)$ with $r=\Omega\left(\sqrt{\frac{\log n}{n}}\right)$, the resistance of any $G$-conformant Markov chain with
$\pi(v)=\Theta\left(\frac{1}{n}\right)$, $\forall v\in V$ satisfies the following with high probability: a) the
conductance $\Phi(\mathbf{P})=O(r)$, and b) the resistance
$R(\mathbf{P})=\Omega(r^{-1})$.
\end{lem}

\begin{proof}
Consider dividing the square with a line parallel to one of its
sides into two halves $S$ and $\bar{S}$ such that $\pi(S)>1/4$  and
$\pi(\bar{S})> 1/4$, as illustrated in Fig. \ref{conduct}. Note that
such a line always exists and needs not to be at the center of the
square. A node in $S$  must lie in the shadowed region to have a
neighbor in $\bar{S}$. For any such node $i$,
$\sum_{j\in\bar{S}}P_{ij}\leq 1$. Applying the Chernoff
bound\cite{Chernoff}, it can be shown that when
$r=\Omega\left(\sqrt{\frac{\log n}{n}}\right)$, the number of nodes
in the shadowed area is upper bounded by $2rn$ w.h.p. Therefore, we
have
\begin{eqnarray}
\Phi(\mathbf{P})<\frac{Q(S,\bar{S})}{\pi(S)\pi(\bar{S})}\leq\frac{2rn\cdot
\Theta\left(\frac{1}{n}\right)\cdot 1}{0.25\cdot 0.25}=\Theta(r),
\end{eqnarray}
i.e., $\Phi(\mathbf{P})=O(r)$ w.h.p. By the the max-flow min-cut
theorem\cite{Sinclair, Leighton}, the resistance $R$ is related to
the conductance $\Phi$ as $R\geq \frac{1}{\Phi}$, thus we have
$R(\mathbf{P})=\Omega(r^{-1})$ w.h.p.
\end{proof}

\begin{figure} \centering
\includegraphics[width=2.0in, bb=212 273 471
538]{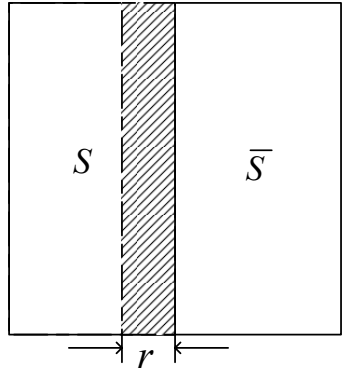}\caption{Upper bound for the conductance of a Markov
chain on $G(n,r)$}\label{conduct}
\end{figure}

Note that the resistance cannot be reduced by lifting \cite{Chen}.
Combining this fact with Lemma \ref{resistancefill} and Lemma
\ref{lowerresistance} yields the following.

\begin{thm} \label{geomix}  Consider a
chain $\mathbf{P}$ on the geometric random graph $G(n,r)=(V,E)$ with
$r=\Omega\left(\sqrt{\frac{\log n}{n}}\right)$ and
$\pi(v)=\Theta\left(\frac{1}{n}\right)$, $\forall v\in V$. For any chain $\mathbf{\tilde{P}}$ lifted from $\mathbf{P}$ and any constant $0<c<1$, $T_{\mathrm{fill}}(\mathbf{\tilde{P}},c)=\Omega(r^{-1})$ with high probability. %
\end{thm}

The above shows that the constructed chain in LADA is optimal in the
scaling law for the mixing parameter $T_{\mathrm{fill}}$ for any
chains lifted from one with an approximately uniform stationary
distribution on $G(n,r)$.

\section{Cluster-based LADA Algorithm for Wireless Networks}
In Section IV-C, we have presented a centralized algorithm, where
the linear iteration is performed on the 2-d grid obtained by
tessellating the geometric random graph. Only the cluster-heads are
involved in the message exchange. Therefore, compared to the purely
distributed LADA algorithm, the centralized algorithm offers an
additional gain in terms of the message complexity, which translates
directly into power savings for sensor nodes. However, as we have
mentioned previously, the assumption of a central controller with
knowledge of global coordinates might be unrealistic. This motivates
us to study a more general cluster-based LADA (C-LADA) algorithm
which alleviates such requirements, and still reaps the benefit of
reduced message complexity.

\subsection{C-LADA Algorithm}
The idea of C-LADA can be described as follows. The nodes are first
clustered using a distributed clustering algorithm given in Appendix
D, where no global coordinate information is required. Two clusters
are considered adjacent (or neighbors) if there is a direct link
joining them. Assume that through some local information exchange, a
cluster-head knows all its neighboring clusters. In the case that
two clusters are joined by more than one links, we assume that the
cluster-heads of both clusters agree on one single such link being
activated. The end nodes of active links are called gateway nodes.
The induced graph $\tilde{G}$ from clustering is a graph with the
vertex set consisting of all cluster-heads and the edge set obtained
by joining the cluster-heads of neighboring clusters. In Fig.
\ref{clusters}, we illustrate the induced graph as a result of
applying our distributed clustering algorithm to a realization of
$G(300, r(300))$, where $r(n)=\sqrt{\frac{2\log n}{n}}$.

\begin{figure} \centering
\includegraphics[width=3.0in]{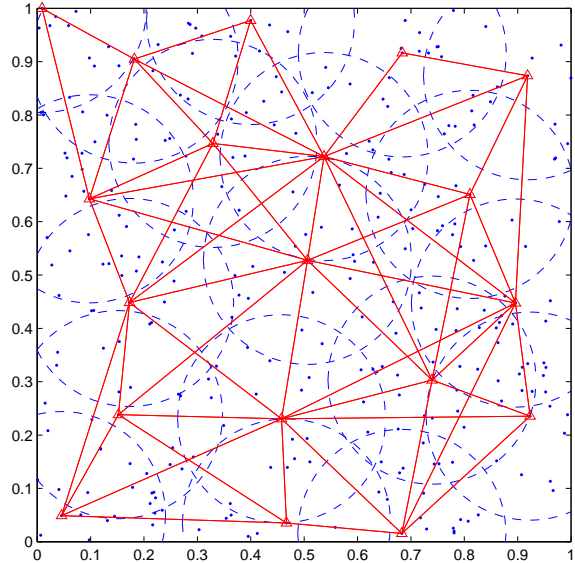}\caption{Illustration of the induced graph from
distributed clustering of a realization of $G(300, r(300))$. Nodes
are indicated with small dots, cluster-heads are indicated with
small triangles, cluster adjacency are indicated with solid lines,
and the transmission range (not clusters) of cluster-heads are
indicated with dashed circles. } \label{clusters}
\end{figure}

As can be seen, the induced graph typically has an arbitrary
topology. Neighbor classification on the induced graph is based on
the relative location of the cluster-heads, according to a similar
rule as described in Section V-A.
Let $\mathcal{N}_m^l$ denote the set of type-$l$ neighboring
clusters (including virtual neighbors) for cluster $m$, and
$d_m^l=|\mathcal{N}_m^l|$.
It can be shown that $d_m^l\geq 1$ for any $m$ and $l$ w.h.p.. Let
$C_i$ be the index of the cluster node $i$ belongs to, and $n_m$ be
the number of nodes in  cluster $m$. It is convenient to consider
another relevant graph $\hat{G}=(V, \hat{E})$ constructed from the
original network graph $G=(V, E)$ as follows: for any $i, j \in V$,
$(i,j)\in \hat{E}$ if and only if $C_i$ and $C_j$ are neighbors.
Moreover, $j$ is considered as a type-$l$ neighbor of $i$ if and
only if $C_j$ is a type-$l$ neighboring cluster of $C_i$. It is easy
to see that nodes in the same cluster have the same set of type-$l$
neighbors, and hence they would follow the same updating rule if the
LADA algorithm is applied. Furthermore, nodes in the same cluster
would have the same values at any time, if their initial values are
the same. Note that the initial values in a given cluster can be
made equal through a simple averaging at the cluster-head. The above
allows updating a cluster as a whole at the cluster-head, saving the
transmissions of individual nodes. For any cluster $m$, let
$\hat{d}_m^l=\sum_{m'\in \mathcal{N}_m^l}n_{m'}$ be the total number
of nodes in the type-$l$ neighboring clusters of $m$, which is equal
to the number of type-$l$ neighbors of any node in cluster $m$ in
$\hat{G}$.

Every cluster-head maintains four pairs of values $(y_m^l, w_m^l)$,
$l=0,\cdots,3$, initialized with $y_m^{l}(0)=\sum_{C_i=m}
x_i(0)/(4n_m)$, and $w_m^{l}(0)=1/4$, $l=0,\cdots,3$. At time $t$,
the gateways nodes of neighboring clusters exchange values and
forward the received values to the cluster-heads. The cluster-head
of cluster $m$ updates its east $y$ value according to
\begin{eqnarray}\label{CLADA}
y_m^{0}(t+1)=\sum_{m'\in
\mathcal{N}_m^2}\frac{n_{m'}}{\hat{d}_{m'}^0}\left[(1-p)y_{m'}^{0}(t)+\frac{p}{2}\left(y_{m'}^{1}(t)+y_{m'}^{3}(t)\right)\right],
\end{eqnarray}
and similarly for other $y$ values and $w$ values, and broadcasts
them to its members. Every node computes the estimate of the average
with
$x_i(t)=\left(\sum_{l=0}^3y_{C_i}^l(t)\right)/\left(\sum_{l=0}^3w_{C_i}^l(t)\right)$.

It can be verified that, the above C-LADA algorithm essentially
realizes the LADA algorithm on graph $\hat{G}$ with the above
neighbor classification rule; for any node in cluster $m$, the
update rule in (\ref{CLADA}) is equivalent to the update rule in
(\ref{LADA-I}). It follows that $\mathbf{x}(t)$ converges to
$x_{\mathrm{ave}}\mathbf{1}$ as $t\rightarrow\infty$, and C-LADA
also achieves an $\epsilon$-averaging time of
$O(r^{-1}\log(\epsilon^{-1}))$ on geometric random graphs.

\subsection{Message Complexity}
Finally, we demonstrate that C-LADA considerably reduces the message
complexity, and hence the energy consumption. For LADA, each node
must broadcast its values during each iteration, hence the number of
messages transmitted in each iteration is $\Theta(n)$. For C-LADA,
there are three types of messages: transmissions between gateway
nodes, transmissions from gateway nodes to cluster-heads and
broadcasts by cluster-heads. Thus, the number of messages
transmitted in each iteration is on the same order as the number of
gateway nodes, which is between $Kd_{\min}$ and $Kd_{\max}$, where
$K$ is the number of clusters, and $d_{\min}$ and $d_{\max}$ are
respectively the maximum and the maximum number of neighboring
clusters in the network.

\begin{lem}\label{neighborsize} Using the Distributed Clustering Algorithm in Appendix D, the number of
neighboring clusters for any cluster $m$ satisfies $4\leq d_m\leq
48$, and the number of clusters satisfies $\pi^{-1}r^{-2}\leq K\leq
2r^{-2}$.
\end{lem}

\begin{proof} The lower bound $d_m\geq 4$ follows from $d_m^l\geq 1$ for any $m$ and $l$.
Note that the cluster-heads are at least at a distance $r$ from each
other (see Appendix D). Hence, the circles with the cluster-heads as
the centers and radius $0.5r$ are non-overlapping. Note also that,
for a cluster $m$, the cluster-heads of all its neighboring clusters
must lie within distance $3r$ from the cluster-head of $m$. Within
the neighborhood of radius $3.5r$ of a cluster-head, there are no
more than $\left(\frac{3.5}{0.5}\right)^2$ non-overlapping circles
of radius $0.5r$. This means that the number of neighboring clusters
is upper bounded by 48.

Consider the tessellation of the unit square into squares of side
$\frac{r}{\sqrt{2}}$. Thus, every such square contains at most one
cluster-head, so there are at most $2r^{-2}$ clusters. On the other
hand, in order to cover the whole unit square, there must be at
least $\pi^{-1}r^{-2}$ clusters.
\end{proof}

The theorem below on the message complexity follows immediately.

\begin{thm}\label{messagecompl} The $\epsilon$-message complexity, defined as the total number of
messages transmitted in the network to achieve $\epsilon$-accuracy,
is $O(nr^{-1}\log(\epsilon^{-1}))$ for the LADA algorithm, and
$O(r^{-3}\log(\epsilon^{-1}))$ for the C-LADA algorithm with high
probability in the geometric random graph $G(n,r)$ with $r =
\Theta(\sqrt{\log n/n})$.
\end{thm}

As a side note, cluster-based algorithms haven also been designed
based on reversible chains\cite{Li_ICASSP07} to reduce the message
complexity.

\section{Related Works}
In this section, we review several relevant works reflecting recent
development on distributed consensus. The reader is referred to
\cite{Bertsekas} for a systematic treatment of distributed
computation. Xiao and Boyd \cite{Xiao} derived necessary and
sufficient conditions for the deterministic weight matrix
$\mathbf{W}$ such that the linear iteration
$\mathbf{x}(t+1)=\mathbf{W}\mathbf{x}(t)$ asymptotically computes
$x_{\mathrm{ave}}\mathbf{1}$ as $t\rightarrow\infty$. They
formulated the fastest linear averaging problem as a semi-definite
program, which is convex when $\mathbf{W}$ is restricted to be
symmetric. Finding the optimal symmetric $\mathbf{W}$ with
non-negative weights is closely tied to the problem of finding the
fastest mixing reversible Markov chain on the graph. Recently,
another class of distributed consensus algorithms, the gossip
algorithms have received much interest\cite{Karp,
Kempe},\cite{Boyd_TIT}. Under the gossip constraint, a node can
communicate with at most one node at a time. In particular, the
randomized gossip algorithm studied by Boyd \emph{et
al.}\cite{Boyd_TIT} realizes distributed averaging through
asynchronous pairwise relaxation.
On a geometric random graph with transmission radius
$\Theta\left(\sqrt{\log n/n}\right)$, the time complexity and
message complexity to reach $\epsilon$-accuracy are respectively
$\Theta\left(n\log\epsilon^{-1}/\log n\right)$ and
$\Theta\left(n^2\log\epsilon^{-1}/\log n\right)$. A recent work by
Moalleimi and Roy \cite{Moalleimi} proposed consensus propagation, a
special form of Gaussian belief propagation, as an alternative for
distributed averaging. By avoiding passing information back to where
it is received, consensus propagation suppresses to some extent the
diffusive nature of a reversible random walk.
However, the gain of consensus propagation in time complexity over
gossip algorithms quickly diminishes as the average node degrees
grow,
in which case the diffusive behavior is not effectively reduced. In
comparison, our LADA algorithms realize distributed consensus with
time complexity $O\left(n^{0.5}\log\epsilon^{-1}/\sqrt{\log
n}\right)$ and message complexity as low as
$O\left(n^{1.5}\log\epsilon^{-1}/(\log n)^{1.5}\right)$ on a
connected geometric random graph.

While the above works studied either synchronous or asynchronous
parallel algorithms, the work by Savas \emph{et al.}\cite{Savas}
explored distributed computation of decomposable functions through
sequential algorithms, where a node does not transmit messages until
it is activated by another node. They proposed two algorithms,
SIMPLE-WALK and COALESCENT, with which the transmission tokens
follow a simple and a coalescing random walk respectively.
Both algorithms provide gain in message complexity at a cost of time
complexity compared with gossip algorithms. The geographic gossip
algorithm proposed by Dimakis \emph{et al.} \cite{Dimakis} is
another work along this line. Motivated by the observation that
standard gossip algorithms can lead to a significant energy waste by
repeatedly circulating redundant information, the geographic gossip
algorithm reduces the message complexity by greedy geographic
routing, for which an overlay network is built so that every pair of
nodes can communicate.
Note that such a modification entails the absolute location
(coordinates) knowledge of the node itself and its neighbors
\footnote{On the contrary, our algorithm only requires direction
knowledge of neighbors.}.  A notable recent work by B\'en\'ezit
\emph{et al.} \cite{Benezit} further improves the geographic gossip
algorithm by allowing averaging along routing paths. Under the
box-greedy routing scheme they propose, further reduction in time
and message complexity is achieved.
Both time and message complexity of the algorithms in \cite{Benezit}
are essentially $\Omega(n\log \epsilon^{-1})$ on geometric random
graphs. In comparison, the class of LADA algorithms we propose
reduce time complexity by a factor of $O\left(\sqrt{n \log
n}\right)$ and increase message complexity by a factor of
$O\left(\sqrt{n}/(\log n)^{1.5}\right)$ to $O(\sqrt{n/(\log n)})$,
and does not require global coordination. The optimal tradeoff
between time and message complexity of distributed consensus
warrants further study.

The independent work by Jung and Shah\cite{Jung} also explored
nonreversible chains for fast distributed consensus. However, our
scheme is considerably different from theirs. Their algorithm adopts
the nonreversible lifting of an existing Markov chain as proposed in
\cite{Chen}, which is constructed from a multi-commodity flow of the
chain with minimum congestion. For each path in the multi-commodity
flow (at least one path between each ordered pair of nodes), a new
replica node (state) is created for each internal node of the path.
Therefore, the state space of the new chain is of a size up to
$n^3$. Moreover, to construct the chain each node in the network
must have global knowledge of the network -- in particular, the
paths in the optimal multi-commodity flow that pass through itself.
On the other hand, the chain used in our algorithm is formed in a
distributed fashion exploiting only local information and simple
computation, and the size of the state space is linear in $n$. As a
result, our algorithm is more robust to topology changes: when a
node joins or leaves the network, only its neighbors need to update
their local processing rules. Therefore, the class of LADA
algorithms we propose is more suited for distributed implementation
in dynamic large-scale networks.

\section{Conclusion}
We propose a class of Location-Aided Distributed Averaging (LADA)
algorithms for grid networks and wireless networks, which achieve
fast convergence via constructing nonreversible lifting of Markov
chains. Our algorithms can realize an $\epsilon$-averaging time of
$O(r^{-1}\log(\epsilon^{-1}))$ for all transmission range $r$ that
guarantees network connectivity, a significant improvement over
existing algorithms based on reversible chains. The cluster-based
LADA (C-LADA) variant requires no central controller to perform
clustering, while reaps the benefit of reduced message complexity.
Our constructed chain attains the optimal scaling law in terms of an
important mixing metric, the fill time \cite{Lovasz}, among all
chains lifted from one with an approximately uniform stationary
distribution on geometric random graphs.


\appendix
\subsection{Proof of Lemma \ref{gridmix}}\label{appgrid}
We will show that by time $t=6k$, the random walk starting from any
state visits every state with probability at least $\frac{C}{4k^2}$
for some constant $C>0$. The desired result then follows from Lemma
\ref{fillmix}. 
Recall  that in Section IV. B., we denote each state $s\in \mathcal{S}$ by a triplet $s=(x, y, l)$.
To facilitate the analysis, we define an auxiliary parameter $z$ for a state $s$ as follows:
\begin{eqnarray}z\triangleq\left\{\begin{array}{cc}
            x & l=\mathrm{E}\\
            2k-x-1 & l=\mathrm{W}\\
            y & l=\mathrm{N}\\
            2k-y-1 & l=\mathrm{S}.
            \end{array}\right.
            \end{eqnarray}
For example, the numbering for east and west states in a given row
is illustrated in Fig. \ref{EW}. Due to the circular numbering, a
horizonal movement of the random walk that keeps the direction (and
bounces back at the boundary) can be written as $(y,z)\rightarrow
(y,z+1 ~(\mathrm{mod}~2k))$, and similarly for a vertical movement.
Note that by defining the function
\begin{eqnarray}
g(z)=\min(z,2k-z-1),
\end{eqnarray}
we have $g(z)=x$ when $l\in\{\mathrm{E,W}\}$, and
$g(z)=y$ when $l\in\{\mathrm{N,S}\}$. 

\begin{figure} \centering
\includegraphics[width=5.0in, bb=10 270 420
400]{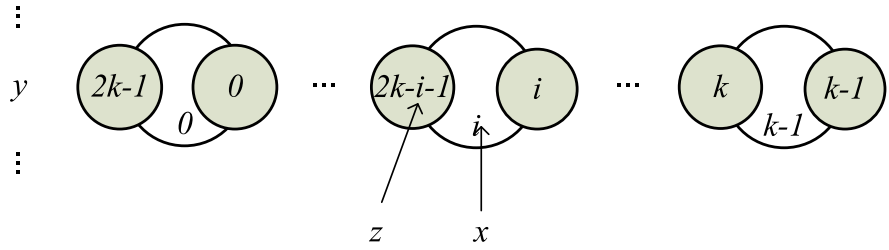}\caption{Illustration of circular numbering of east and west
states within a row}\label{EW}
\end{figure}

Without loss of generality, we
assume that the chain starts from some horizontal state $s_0=(x_0, y_0, l_0)$ with $l_0\in\{\mathrm{E,W}\}$.
Let $T_1, T_2, \cdots,$ ($1\leq T_1\leq T_2\leq\cdots$) be the
times that the random walk makes a turn. Let $s_t$ be the state the
random walk visits at the $t$th step\footnote{In our notation, in the $t$th step, the random walk goes from state $s_{t-1}$ to $s_{t}$.}, and $A_t$ be the number of
turns made by the random walk up to time $t$. In the following, we consider two cases: (1) a target state $s = (x,y,l)$ with $l \in\{\mathrm{E,W}\}$, i.e., a horizontal state, and (2) a target state with $l \in \{\mathrm{N,S}\}$, i.e., a vertical state, and show that at $t=6k$, for both cases
\begin{eqnarray*} \Pr\{s_t=s\}
\geq 
\frac{C}{4k^2}.
\end{eqnarray*}

\begin{figure} \centering
\includegraphics[width=6.0in, bb=11 87 577 503]{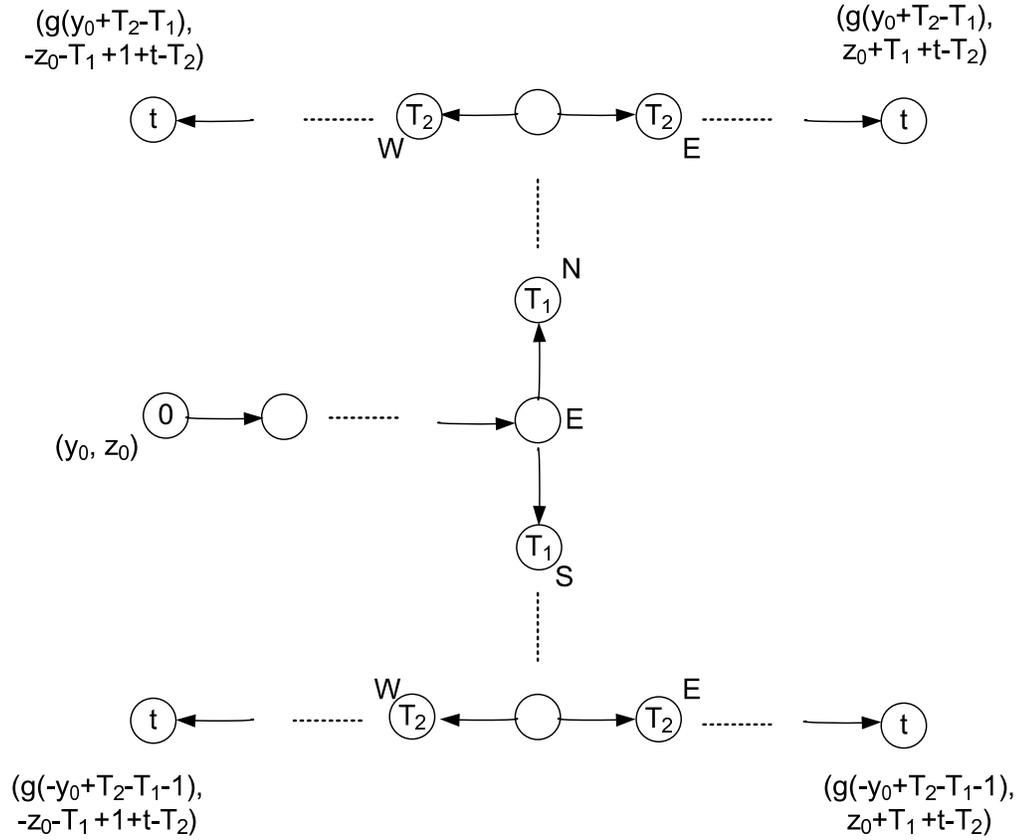}\caption{Illustration of states traversed till time $t$ with two turns}\label{randomwalk_grid}
\end{figure}

\begin{enumerate}
\item \textbf{$s$ is a horizontal state}. In this case, we focus on $A_t=2$ (so $s_t$ is also a horizontal state), and show that \begin{eqnarray}
    \Pr\{s_t=s\} \geq \Pr\{s_t=s, A_t=2\} \geq \frac{C}{4k^2}.
    \end{eqnarray}
Note that a horizontal state $s$ is fully characterized by $y$ and $z$ (since $x=g(z)$). Thus, the state at time 0 can be represented as $(y_0, z_0)$, as illustrated in Fig. \ref{randomwalk_grid}. Now, consider the state at time $t$. First, observe that $y_t$ is determined only by the direction of the first turn at $T_1$, which may be towards north or south, as illustrated by the two states labeled with $T_1$ in Fig. \ref{randomwalk_grid}. If the turn is towards north, we have
\begin{eqnarray}\label{A2s11}
y_t = g(y_0 + T_2-T_1 ~(\mathrm{mod}~2k));
\end{eqnarray}
if it is towards south, we have
\begin{eqnarray}\label{A2s12}
y_t = g(2k-1-y_0+T_2-T_1 ~(\mathrm{mod}~2k)) = g(-y_0+T_2-T_1-1 ~(\mathrm{mod}~2k)) .
\end{eqnarray}
Second, observe that $z_t$ is determined only by the direction of the second turn at $T_2$, which may be the same as the one in which the random walk is moving at time $T_1-1$, or the opposite.  In the former case (the two east states at time $T_2$ shown in Fig. \ref{randomwalk_grid}), it can be shown (by observing the two periods $[1, T_1-1]$ and $[T_2, t]$ within which the random walk is traveling horizontally) that
\begin{eqnarray}\label{A2s21}
z_t = z_0 + T_1-1 + (t-T_2+1) ~(\mathrm{mod}~2k) = z_0+T_1-T_2+t ~(\mathrm{mod}~2k);
\end{eqnarray}
in the latter case  (the two west states at time $T_2$ shown in Fig. \ref{randomwalk_grid}), we have
\begin{eqnarray}\label{A2s22}
z_t = 2k-1-(z_0 + T_1-1) + (t-T_2+1) ~(\mathrm{mod}~2k) = -z_0-T_1-T_2+t+1 ~(\mathrm{mod}~2k).
\end{eqnarray}
Therefore, we have at $t=6k$,
\begin{eqnarray*}
&&\Pr\{s_t=s\}\geq\Pr\{s_t=s, A_t = 2\} \\& \geq &\Pr\{g(y_0 + T_2-T_1) =
y ~(\mathrm{mod}~2k),\quad -z_0-T_1-T_2+t+1 = z
~(\mathrm{mod}~2k),~ A_t=2\}
\\ &&+ \Pr\{g(-y_0+T_2-T_1-1) = y~(\mathrm{mod}~2k),\quad -z_0-T_1-T_2+t+1 = z ~(\mathrm{mod}~2k),~
A_t=2\},
    \end{eqnarray*}
    where the second inequality comes from picking two combinations of $y_t$ and $z_t$ out of
            the four possible combinations formed from (\ref{A2s11}) - (\ref{A2s22}). Assuming that
            $g(i)=i$ (the case for $g(i)=2k-1-i$ can be similarly argued), and letting $a=y-y_0$,
            $b=t-z_0-z+1$ and $c=y+y_0+1$, we get
\begin{eqnarray*}
\Pr\{s_t=s\}&\geq& \Pr\{T_2-T_1 = a~(\mathrm{mod}~2k),\quad T_1+T_2
= b ~(\mathrm{mod}~2k),~ A_t=2\}
\\ &&+ \Pr\{T_2-T_1 = c ~(\mathrm{mod}~2k),\quad T_1+T_2 =b~(\mathrm{mod}~2k),~
A_t=2\}.
\end{eqnarray*}
Note that $T_2-T_1$ and $T_1+T_2$ must have the same parity, so we
need to consider two cases: if $a$ and $b$ have the same parity,
then there exists at least a pair of $(T_1,T_2)$ with $1\leq
T_1<T_2\leq t$ (e.g., $T_1=\frac{b-a}{2}-1 ~(\mathrm{mod}~2k)+1$ and
$T_2=\frac{a+b}{2}-1 ~(\mathrm{mod}~2k)+2k+1$) such that $T_2-T_1 =
a~(\mathrm{mod}~2k)$ and $T_1+T_2 = b ~(\mathrm{mod}~2k)$ are
satisfied; if $a$ and $b$ have different parities, then $c$ and $b$
must have the same parity, and there exists at least a pair of
$(T_1,T_2)$ with $1\leq T_1<T_2\leq t$ such that the second set of
equations above is satisfied. Either of the two cases occurs with a
probability $\frac{1}{4k^2}\left(1-\frac{1}{k}\right)^{t-2}$. Using
the fact that $(1-\frac{1}{k})^k\geq 1/4$ for $k>2$, at $t=6k$ we
get
\begin{eqnarray}
\Pr\{s_t=s\}\geq \frac{1}{4k^2}\left(1-\frac{1}{k}\right)^{t-2}>
\frac{2^{-12}}{4k^2}.
\end{eqnarray}

\item \textbf{$s$ is a vertical state}. We show that in this case it is sufficient to consider the case of $A_t=3$.
Similarly as above, a vertical state $s$ is fully characterized by $x$ and $z$. 
Note that $x_t$ is only determined by the direction of the second turn. Similar to (\ref{A2s21}) and (\ref{A2s22}) two possible values for $x_t$ are given by
\begin{eqnarray} x_t = \left\{\begin{array}{ll}\label{A3s1}
            g(z_0+T_1-T_2+T_3-1 ~(\mathrm{mod}~2k)) \\
            g(-z_0-T_1-T_2+T_3 ~(\mathrm{mod}~2k)).
            \end{array}\right.
            \end{eqnarray}
Also $z_t$ is only determined by the direction of the first turn and third turn. It can be shown that the four possible values of $z_t$ are given by
\begin{eqnarray}\small\label{A3s2}
z_t=\left\{\begin{array}{ll}
            y_0+t-T_1+T_2-T_3+1 ~(\mathrm{mod}~2k) \\
            -y_0+t+T_1-T_2-T_3 ~(\mathrm{mod}~2k) \\
            -y_0+t-T_1+T_2-T_3 ~(\mathrm{mod}~2k) \\
            y_0+t+T_1-T_2-T_3+1 ~(\mathrm{mod}~2k).
            \end{array}\right.
            \end{eqnarray}
            Therefore,
\begin{eqnarray}
&&\Pr\{s_t=s\}\geq\Pr\{s_t=s, A_t = 3\}\nonumber\\
&\geq&\Pr\{z_0+T_1-T_2+T_3-1 = x ~(\mathrm{mod}~2k), ~
y_0+t+T_1-T_2-T_3+1 = z ~(\mathrm{mod}~2k),~ A_t=3\}\nonumber
\\ &&+\Pr\{z_0+T_1-T_2+T_3-1 = x
~(\mathrm{mod}~2k),~ -y_0+t+T_1-T_2-T_3 = z
~(\mathrm{mod}~2k),~ A_t=3\}\nonumber\\
            &=&\Pr\{T_3-(T_2-T_1) = a ~(\mathrm{mod}~2k), \quad
T_3+(T_2-T_1) = b ~(\mathrm{mod}~2k),~ A_t=3\} \label{At31}\\
&&+\Pr\{T_3-(T_2-T_1) = a ~(\mathrm{mod}~2k),~ T_3+(T_2-T_1) = c
~(\mathrm{mod}~2k),~ A_t=3\}\label{At32},
\end{eqnarray}
where the second inequality comes from picking two combinations out
of eight possible combinations formed from (\ref{A3s1}) and
(\ref{A3s2}), and in the last inequality, we have substituted
$a=x-z_0+1$, $b=y_0+t-z+1$ and $c=-y_0+t-z$. Same as 1), we must
consider two cases on parity. For $a$ and $b$ with the same parity,
consider the $2k$ triplets of $(T_1, T_2, T_3)$ given by
\begin{eqnarray*}
\left(T_1,~\frac{b-a}{2}-1~(\mathrm{mod}~2k)+1+T_1,~
\frac{b+a}{2}-1~(\mathrm{mod}~2k)+1+4k\right), \quad T_1=1,2,\cdots
2k.
\end{eqnarray*}
It is obvious that any such triplet satisfies $1\leq T_1<T_2<T_3\leq
6k$, as well as the conditions in (\ref{At31}). For $a$ and $b$ with
different parity, $a$ and $c$ must have the same parity, and
similarly there exists at least $2k$ valid triplets of $(T_1, T_2,
T_3)$ satisfying the conditions in (\ref{At32}). Thus, for any
target vertical state $s$, we can always find $2k$ turning times
$(T_1, T_2, T_3)$ with proper turning directions to reach $s$ at
$t=6k$ with probability
\begin{eqnarray}
\Pr\{s_t=s\}\geq
2k\cdot\frac{1}{8k^3}\left(1-\frac{1}{k}\right)^{t-3}>
\frac{2^{-12}}{4k^2}.
\end{eqnarray}
\end{enumerate}

This completes the proof.

\subsection{Proof of Lemma \ref{LADAmix}}\label{appLADA}
Assume the unit square is coordinated by $(x, y)$ with $x,
y\in[0,1]$, starting from the south-west corner. Denote the state
space of the chain $\tilde{\mathbf{P}}_1$ by $\mathcal{S}$. A state
$s\in \mathcal{S}$ is represented with a triplet $s=(x, y, l)$
following the grid case in Appendix \ref{appgrid}.  Define an
auxiliary parameter $z$ for a state $s$ as follows:
\begin{eqnarray*}z\triangleq\left\{\begin{array}{cc}
            x & l=\mathrm{E}\\
            2-x & l=\mathrm{W}\\
            y & l=\mathrm{N}\\
            2-y & l=\mathrm{S}.\\
            \end{array}\right.
            \end{eqnarray*}
We will show that by the time $t=6k+1$, for any state $s\in
\mathcal{S}$, $\Pr\{s_t = s\} \geq c_1\pi(s)$ for some positive
constant $c_1$.

Consider a movement of the random walk. Denote the distance traveled
in the direction of movement, and that orthogonal to the direction
of movement at time $t$ respectively by $\alpha_t$ and $\beta_t$, as
shown in Fig. \ref{movement}. Since nodes are randomly and uniformly
distributed and the transition probability is uniform for all
neighbors in the same direction, we can calculate the expected value
of $\alpha_t$ and $\beta_t$ (with respect to the node distribution)
as follows:
\begin{eqnarray}
\mathds{E}(\alpha_t)&=&\frac{4}{\pi
r^2}\int_{-\pi/4}^{\pi/4}\int_0^r x^2 \cos\theta ~dx~
d\theta=\frac{4\sqrt{2}}{3\pi}r\triangleq
\mu_\alpha,\\
\mathds{E}(\beta_t)&=&\frac{4}{\pi r^2}\int_{-\pi/4}^{\pi/4}\int_0^r
x^2 \sin\theta ~dx~ d\theta=0.
\end{eqnarray}
Similarly, their second-order moments can be readily computed as
\begin{eqnarray}
\mathds{E}(\alpha_t^2)&=&\frac{4}{\pi
r^2}\int_{-\pi/4}^{\pi/4}\int_0^r x^3 \cos^2\theta ~dx~
d\theta=\frac{\pi+\sqrt{2}}{4\pi}r^2,\\
\mathds{E}(\beta_t^2)&=&\frac{4}{\pi
r^2}\int_{-\pi/4}^{\pi/4}\int_0^r x^3 \sin^2\theta ~dx~
d\theta=\frac{\pi-\sqrt{2}}{4\pi}r^2,
\end{eqnarray}
and the variances of $\alpha_t$ and $\beta_t$ are given by
\begin{eqnarray}
\left(\frac{\pi+\sqrt{2}}{4\pi}-\frac{32}{9\pi^2}\right)r^2&\triangleq&
\sigma_\alpha^2,\\
\frac{\pi-\sqrt{2}}{4\pi}r^2&\triangleq&\sigma_\beta^2.
\end{eqnarray}
Note that $\alpha_t$ and $\beta_t$ are uncorrelated, i.e.,
\begin{eqnarray}
\mathds{E}((\alpha_t-\mu_\alpha)\beta_t)=\mathds{E}(\alpha_t\beta_t)=\frac{4}{\pi
r^2}\int_{-\pi/4}^{\pi/4}\int_0^r x^3 \cos\theta \sin\theta ~dx~
d\theta=0.
\end{eqnarray}
In the following, we assume
$k=\ulcorner\frac{1}{\mu_\alpha}\urcorner$ and the turning
probability $p=\frac{1}{k}=\Theta(r)$.

\begin{figure} \centering
\includegraphics[width=2.5in, bb=200 320 420
550]{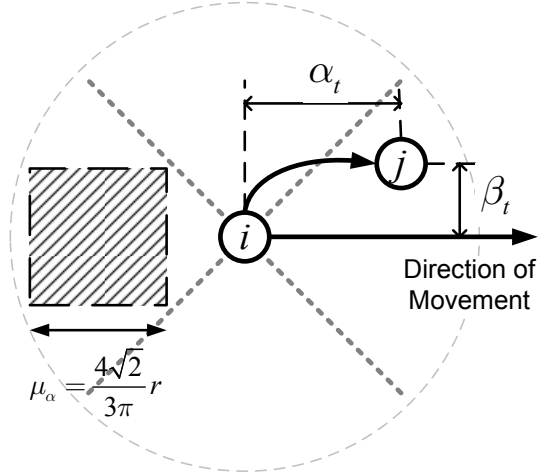}\caption{Illustration of moving distances and target
set}\label{movement}
\end{figure}

Without loss of generality, we assume that the random walk starts
from some arbitrary horizontal state $s_0 = (x_0, y_0, l_0)$ with
$l_0\in\{\mathrm{E,W}\}$, $y_0=a_0\mu_\alpha$ for some
$a_0\in\{0,1,\cdots,k-1\}$  and the corresponding
$z_0=b_0\mu_\alpha$ for some $b_0\in\{0,1 \cdots,
2k-1\}$.\footnote{Recall that a horizontal node is completely
characterized by $y$ and $z$. The proof is essentially the same for
non-integer $a_0$ and $b_0$, with a little more complicated
notation.} Similar to Appendix \ref{appgrid}, we need to consider
two cases: the target state $s$ being a horizontal state and the
target state $s$ being a vertical state. In the following, we will
focus on the the former case, and the proof for the latter case is
similar.

First consider the \emph{expected} location $\mathds{E}(s_t)$ of the
random walk at $t$. It depends only on the turning times and turning
directions, and evolves according to the random walk
$\tilde{\mathbf{P}}$ on the $k\times k$ grid (see Section IV)
\footnote{If $p=\frac{c}{k}$ for some positive $c\neq 1$, then the
expected location would evolve according to another chain which
differs from $\tilde{\mathbf{P}}$ only in the turning probability,
and has the same scaling law in the mixing time as
$\tilde{\mathbf{P}}$.}. Thus, according to Appendix \ref{appgrid},
at $t=6k$, for any $a'\in\{0,1,\cdots,k-1\}$ and
$b'\in\{0,1,\cdots,2k-1\}$, we have
\begin{eqnarray}\label{bound6k}
\Pr\{\mathds{E}(y_t)=a'\mu_\alpha,
\mathds{E}(z_t)=b'\mu_\alpha\}\geq\Pr\{\mathds{E}(y_t)=a'\mu_\alpha,
\mathds{E}(z_t)=b'\mu_\alpha, A_t=2\}\geq\frac{c_2}{4k^2}
\end{eqnarray} for some $c_2>0$.

In order to obtain a lower bound for the probability of reaching a
target horizontal state $s$ at $t=6k+1$, we first obtain a lower
bound for the probability of reaching any ancestor of $s$ in the
underlying graph of the chain at $t=6k$. For example, consider an
east state $s$ of node $i$ as in Fig. \ref{movement}. Note that the
effective west neighboring region of node $i$ covers a circular
sector of 90 degrees (for boundary nodes virtual neighbors are
considered). It can be shown that such a circular sector contains a
square of side $\mu_\alpha$ as depicted in Fig. \ref{movement} (for
boundary nodes the corresponding square is folded along the
boundary). Denote the set of east states in
$\mathcal{N}_i^2\bigcup\widehat{\mathcal{N}}_i^2$ and west states in
$\widetilde{\mathcal{N}}_i^2$ in this square by
$\mathcal{\hat{S}}=\{\hat{s}: \hat{y}\in \hat{Y}, ~ \hat{z}\in
\hat{Z},~\hat{l}\in\{\mathrm{E,W}\}\}$, where generally for a
non-boundary node, we have $\hat{Y}=[a\mu_{\alpha},
(a+1)\mu_{\alpha})$ and $\hat{Z}=[b\mu_{\alpha}, (b+1)\mu_{\alpha})$
for some $a\in[0,k-2]$ and $b\in [0, 2k-2]$, and $\hat{l}=l$
 (the direction of the target state)\footnote{In the above example, if $i$ is a west boundary node, then
the square under consideration is folded along the west boundary,
such that $\hat{Z}=[0, (1-b)\mu_{\alpha})\bigcup [2-b\mu_{\alpha},
2)$ for some $b\in(0,1)$, with the latter corresponding to  west
states of nodes in $\widetilde{\mathcal{N}}_i^2$. Note that in all
cases, both $\hat{Y}$ and $\hat{Z}$ consist of intervals with a
total length $\mu_{\alpha}$.}. In the following, we assume $i$ is
not a boundary node for simplicity, but the proof extends easily to
the boundary nodes.

We claim that at $t=6k$,
\begin{equation}\label{auxiliary}
    \sum_{a'=0}^{k-1}\sum_{b'=0}^{2k-1}\Pr\left\{
s_t\in\mathcal{\hat{S}}~|~\mathds{E}(y_t)=a'\mu_\alpha,
\mathds{E}(z_t)=b'\mu_\alpha, A_t=2\right\} \geq c'
\end{equation}
for some constant $c'$ w.h.p. Based on this result and
(\ref{bound6k}), we have at $t=6k$,
\begin{eqnarray}
\Pr\{s_t\in\mathcal{\hat{S}}\}&\geq&
\sum_{a'=0}^{k-1}\sum_{b'=0}^{2k-1}\Pr\left\{
s_t\in\mathcal{\hat{S}}~|~\mathds{E}(y_t)=a'\mu_\alpha,
\mathds{E}(z_t)=b'\mu_\alpha,
A_t=2\right\}\nonumber\\
&&\cdot\Pr\{A_t=2, \mathds{E}(y_t)=a'\mu_\alpha,
\mathds{E}(z_t)=b'\mu_\alpha\}\geq\frac{c'c_2}{4k^2}.
\end{eqnarray}

By Lemma \ref{regularity}, when $r>\sqrt{\frac{16\log n}{\pi n}}$,
$d_{\max}\triangleq\max_{i,l} d_i^l\leq c_3n r^2$ for some constant
$c_3>0$ w.h.p., thus we have
\begin{eqnarray}\label{lb6k}
\Pr\{s_{6k+1} = s \}\geq\frac{1}{2}\sum_{\hat{s}\in
\mathcal{\hat{S}}}\frac{\Pr\{s_{6k} = \hat{s} \}}{d_{\max}}\geq
\frac{1/2}{c_3n r^2}\frac{c'c_2}{4k^2}\triangleq\frac{c_4}{4n}.
\end{eqnarray}

Note that, the random walk $\tilde{\mathbf{P}}$ has a uniform
stationary
distribution on the $k\times k$ grid. 
Using the argument as above, it can be shown that for any set
$\mathcal{\hat{S}}$ containing states of the same type in a square
of side $\mu_{\alpha}$, the stationary probability of
$\tilde{\mathbf{P}}_1$ satisfies
$\pi(\mathcal{\hat{S}})=\frac{1}{4k^2}$, and consequently the
stationary probability of any state of $\tilde{\mathbf{P}}_1$ is
lower bounded by $\frac{c_5}{4n}$ for some $c_5>0$
(c.f.(\ref{lb6k})). For an upper bound, note that in Fig.
\ref{movement} the effective west neighboring region of $i$ is also
contained in an area $A$ consisting of $2\times 3$ squares of side
$\mu_\alpha$. Let $\mathcal{S}^{0}$, $\mathcal{S}^1$ and
$\mathcal{S}^3$ respectively denote the set of east
states\footnote{For nodes in $\widetilde{\mathcal{N}}_i^2$, their
west states are considered instead.}, the set of north states and
the set of south states of (physical and virtual) west neighbors of
$i$ that lie in $A$. By Lemma \ref{regularity}, when
$r>\sqrt{\frac{16\log n}{\pi n}}$, $d_{\min}\triangleq\min_{i,l}
d_i^l\geq c_6nr^2$ w.h.p. Hence for any state $s$,
\begin{eqnarray}
\pi(s)\leq(1-p)\sum_{s\in
\mathcal{S}^0}\frac{\pi(s)}{d_{\min}}+\frac{p}{2}\left[\sum_{s\in
\mathcal{S}^1}\frac{\pi(s)}{d_{\min}}+\sum_{s\in
\mathcal{S}^3}\frac{\pi(s)}{d_{\min}}\right]\leq
\left[(1-p)+\frac{p}{2}\cdot 2\right]\frac{1}{c_6n
r^2}\cdot\frac{6}{4k^2}\triangleq\frac{c_7}{4n}.
\end{eqnarray}

We conclude that the stationary distribution of
$\tilde{\mathbf{P}}_1$ is approximately uniform, i.e., for any $s\in
\mathcal{S}$, $\frac{c_5}{4n} \leq \pi(s) \leq \frac{c_7}{4n}$ for
some $c_5, c_7>0$. It follows from (\ref{lb6k}) that $\Pr\{s_{6k+1}
= s \}\geq \frac{c_4}{c_7}\pi(s)\triangleq c_1\pi(s)$ w.h.p., which
implies that the fill time of $\tilde{\mathbf{P}}_1$ is
$T_{\mathrm{fill}}(\tilde{\mathbf{P}}_1,\epsilon)=O(r^{-1})$ w.h.p.

We are left to verify the claim (\ref{auxiliary}). It is sufficient
to consider the case that the random walk makes two turns in first
$6k$ steps, with the turning times $T_1$ and $T_2$. Denote the
distance vector traveled at the $t$th step by
\begin{eqnarray}\Lambda_t\triangleq\left\{\begin{array}{cc}
                {[\alpha_{t}~~\beta_{t}]}^T & t\in[1, T_1)\cup[T_2, 6k]\\
                {[\beta_{t}~~\alpha_{t}]}^T & t\in[T_1, T_2),\\
                \end{array}\right.
                \end{eqnarray}
with mean
\begin{eqnarray}\mathds{E}(\Lambda_t)\triangleq\mu_\Lambda=\left\{\begin{array}{cc}
                {[\mu_\alpha~~0]}^T & t\in[1, T_1)\cup[T_2, 6k]\\
                {[0~~\mu_\alpha]}^T & t\in[T_1, T_2),\\
                \end{array}\right.
                \end{eqnarray}
and  covariance matrix (note $\alpha_t$ and $\beta_t$ are
uncorrelated)
\begin{eqnarray}\Sigma_\Lambda =\left\{\begin{array}{cc}
                \left[\begin{array}{cc}
                \sigma_\alpha^2 & 0 \\
                0 & \sigma_\beta^2
                \end{array}\right] & t\in[1, T_1)\cup[T_2, 6k]\\[0.5cm]
                \left[\begin{array}{cc}
                \sigma_\beta^2 & 0 \\
                0 & \sigma_\alpha^2
                \end{array}\right] & t\in[T_1, T_2).\\
                \end{array}\right.
                \end{eqnarray}

As the distance vectors in different steps are independent, the
covariance matrix of the total distance vector
$\Lambda=\sum_{t=1}^{6k}\Lambda_{t}$ is given by
\begin{eqnarray}\label{total_cov}
\Sigma_{\Lambda|T_1,T_2} = \left( \begin{array}{cc}
\sigma_{\alpha|T_1,T_2}^2 & 0 \\
0 & \sigma_{\beta|T_1,T_2}^2
\end{array} \right),
\end{eqnarray}
where
\begin{eqnarray}
\sigma_{\alpha|T_1,T_2}^2=[T_1+(6k-T_2)]\sigma_\alpha^2+(T_2-T_1)\sigma_\beta^2=(\sigma_\beta^2-\sigma_\alpha^2)(T_2-T_1)+6k\sigma_\alpha^2
\end{eqnarray}
and
\begin{eqnarray}
\sigma_{\beta|T_1,T_2}^2=[T_1+(6k-T_2)]\sigma_\beta^2+(T_2-T_1)\sigma_\alpha^2=(\sigma_\alpha^2-\sigma_\beta^2)(T_2-T_1)+6k\sigma_\beta^2
\end{eqnarray}
are the respective variance of the total distance traveled
horizontally and vertically in $6k$ steps. As
$\sigma_\beta^2>\sigma_\alpha^2$, it is easy to verify that the
maximum of $\sigma_{\alpha|T_1,T_2}^2$ and
$\sigma_{\beta|T_1,T_2}^2$ (with respect to $T_{1}$ and $T_{2}$) are
the same:
\begin{eqnarray}\label{maxvar}
\sigma_{\alpha,\mathrm{max}}^2=\sigma_{\beta,\mathrm{max}}^2=\sigma_\alpha^2+(6k-1)\sigma_\beta^2.
\end{eqnarray}

Let
\begin{eqnarray}\Lambda_{k,t}\triangleq\Sigma_{\Lambda|T_1,T_2}^{-1/2}(\Lambda_t-\mu_\Lambda)=\left\{\begin{array}{cc}
            \left[\begin{array}{c}
            (\alpha_t-\mu_\alpha)/\sigma_{\alpha|T_1,T_2} \\
            \beta_t/\sigma_{\beta|T_1,T_2}
            \end{array}\right] & t\in[1, T_1)\cup[T_2, 6k]\\[0.5cm]
            \left[\begin{array}{c}
            \beta_t/\sigma_{\alpha|T_1,T_2} \\
            (\alpha_t-\mu_\alpha)/\sigma_{\beta|T_1,T_2}
            \end{array}\right] & t\in[T_1, T_2),\\
            \end{array}\right.
            \end{eqnarray}
we have $\mathds{E}(\Lambda_{k,t})=\mathbf{0}$ and $\lim_{n \to
\infty}\sum_{t=1}^{6k}\mathds{E}(\Lambda_{k,t}\Lambda_{k,t}^T)=\mathbf{I}$,
where $\mathbf{I}$ is the $2\times2$ identity matrix. In addition,
by defining $\mathds{E}(Y;C)=\mathds{E}(Y1_C)$ with $1_C$ being the
indicator function of $C$, for any $\epsilon>0$
\begin{eqnarray}\label{lindebergcondition}
\lim_{n \to \infty}
\sum_{t=1}^{6k}\mathds{E}(|\Lambda_{k,t}|^2;|\Lambda_{k,t}|>\epsilon)=0,
\end{eqnarray}
since $|\Lambda_{k,t}|$ is always less than $\epsilon$ when $n$ is
sufficiently large such that
$\frac{r}{\max\{\sigma_{\alpha|T_1,T_2},~\sigma_{\beta|T_1,T_2}\}}<\epsilon/2$.
Then according to the multivariate Lindeberg-Feller Theorem
(\cite{Vaart} Proposition 2.27), the conditional probability density
function (PDF) of
\begin{eqnarray}
\sum_{t=1}^{6k}\Lambda_{k,t}=\Sigma_{\Lambda|T_1,T_2}^{-1/2}\sum_{t=1}^{6k}(\Lambda_t-\mu_\Lambda)=
            \left[\begin{array}{c}
            (z_{6k}-\mathds{E}(z_{6k}))/\sigma_{\alpha|T_1,T_2} \\
            (y_{6k}-\mathds{E}(y_{6k}))/\sigma_{\beta|T_1,T_2}
            \end{array}\right],
\end{eqnarray}
given $T_1$ and $T_2$ \footnote{which determine $\mathds{E}(z_t)$
and $\mathds{E}(y_t)$ (for fixed turning directions), but not vice
versa. There may exist multiple combinations of $\{T_1,T_2\}$ which
can result in the same $\{\mathds{E}(z_t), \mathds{E}(y_t)\}$.}
converges in distribution to the standard multivariate normal
distribution $\mathcal{N}(0,\mathbf{I})$.

Suppose $\mathcal{T}_{\{a',b'\}}$ is the set of turning times
combination that result in $\mathds{E}(z_t)=b'\mu_\alpha$,
$\mathds{E}(y_t)=a'\mu_\alpha$, and
\begin{equation*}
\{T_{1,\{a',b'\}},T_{2,\{a',b'\}}\}=\mathrm{argmin}_{\{T_1,T_2\} \in
\mathcal{T}_{\{a',b'\}}} \Pr\left\{z_t\in
[b\mu_\alpha,(b+1)\mu_\alpha), y_t\in
[a\mu_\alpha,(a+1)\mu_\alpha)~|~T_1,T_2\right\}
\end{equation*}
for any $a\in[0,k-2]$ and $b\in [0, 2k-2]$.
Define
\begin{eqnarray*}
\Pi(X;\Lambda,\Sigma)=\frac{1}{2\pi\sqrt{|\Sigma|}}\exp\{-\frac{1}{2}(X-\Lambda)^T\Sigma^{-1}(X-\Lambda)\}
\end{eqnarray*}
as the PDF value of the multivariate normal distribution
$\mathcal{N}(\Lambda,\Sigma)$ at $X$, and (c.f. (\ref{total_cov}))
\begin{equation*}
    \Pi'_{\{a',b'\}}(X)=\Pi(X;[b'\mu_\alpha~~a'\mu_\alpha]^T,\Sigma_{\Lambda|T_{1,\{a',b'\}},T_{2,\{a',b'\}}}).
\end{equation*}
Then for any $a\in[0,k-2]$ and $b\in [0,2k-2]$, we can always find a
matrix (c.f. (\ref{maxvar}))
\begin{eqnarray*}
\Sigma_0=\left( \begin{array}{cc}
\sigma_{\alpha_0}^2 & 0 \\
0 & \sigma_{\beta_0}^2
\end{array} \right)
\end{eqnarray*}
satisfying
\begin{align}\label{sigma0}
\frac{1}{2\pi\sqrt{|\Sigma_0|}}\leq\min_{a'=0,...,k-1,b'=0,1,...,2k-1}
\bigg\{
&\Pi'_{\{a',b'\}}([b\mu_\alpha~~a\mu_\alpha]^T),\nonumber\\
&\Pi'_{\{a',b'\}}([(b+1)\mu_\alpha~~a\mu_\alpha]^T),\nonumber\\
&\Pi'_{\{a',b'\}}([b\mu_\alpha~~(a+1)\mu_\alpha]^T),\nonumber\\
&\Pi'_{\{a',b'\}}([(b+1)\mu_\alpha~~(a+1)\mu_\alpha]^T)\big\}
\bigg\}.
\end{align}
This allows us to define an auxiliary normal distribution with an
arbitrary mean and covariance matrix $\Sigma_0$ whose maximal PDF
value is less than the minimum PDF values of all $\Pr\{z_{6k},
y_{6k}~|~\mathds{E}(z_{6k})=b'\mu_\alpha,
\mathds{E}(y_{6k})=a'\mu_\alpha, A_{6k}=2\}$ ($a'=0,...,k-1$,
$b'=0,...,2k-1$) in the square $\{\hat{s}: \hat{z} \in [b\mu_\alpha,
(b+1)\mu_\alpha], \hat{y} \in [a\mu_\alpha, (a+1)\mu_\alpha]\}$.
Therefore, as $n\to\infty$,
\begin{align}\label{lbprob}
&\sum_{a'=0}^{k-1}\sum_{b'=0}^{2k-1}\Pr\left\{
y_{6k}\in[a\mu_\alpha,(a+1)\mu_\alpha),z_{6k}\in[b\mu_\alpha,(b+1)\mu_\alpha)~|~
\mathds{E}(y_{6k})=a'\mu_\alpha, \mathds{E}(z_{6k})=b'\mu_\alpha, A_{6k}=2\right\}\nonumber\\
&\geq\sum_{a'=0}^{k-1}\sum_{b'=0}^{2k-1}\int_{a\mu_\alpha}^{(a+1)\mu_\alpha}\int_{b\mu_\alpha}^{(b+1)\mu_\alpha}
\frac{1}{2\pi\sigma_{\beta|T_{1,\{a',b'\}},T_{2,\{a',b'\}}}\sigma_{\alpha|T_{1,\{a',b'\}},T_{2,\{a',b'\}}}}\nonumber\\
&\hspace{5cm}\exp\bigg\{-\frac{(y_t-a'\mu_\alpha)^2}{2\sigma_{\beta|T_{1,\{a',b'\}},T_{2,\{a',b'\}}}^2}-\frac{(z_t-b'\mu_\alpha)^2}{2\sigma_{\alpha|T_{1,\{a',b'\}},T_{2,\{a',b'\}}}^2}\bigg\}dz_t dy_t\nonumber\\
&\geq\sum_{a'=0}^{k-1}\sum_{b'=0}^{2k-1}\int_{a\mu_\alpha}^{(a+1)\mu_\alpha}\int_{b\mu_\alpha}^{(b+1)\mu_\alpha}
\frac{1}{2\pi \sigma_{\beta_0}\sigma_{\alpha_0}}\exp\bigg\{-\frac{(y_t-a'\mu_\alpha)^2}{2\sigma_{\beta_0}^2}-\frac{(z_t-b'\mu_\alpha)^2}{2\sigma_{\alpha_0}^2}\bigg\}dz_t dy_t\nonumber\\
&=\sum_{a'=0}^{k-1}\int_{(a-a')\mu_\alpha}^{(a+1-a')\mu_\alpha}\frac{1}{\sqrt{2\pi}\sigma_{\beta_0}}\exp\bigg\{-\frac{y_t^2}{2\sigma_{\beta_0}^2}\bigg\}dy_t
\sum_{b'=0}^{2k-1}\int_{(b-b')\mu_\alpha}^{(b+1-b')\mu_\alpha}\frac{1}{\sqrt{2\pi}\sigma_{\alpha_0}}\exp\bigg\{-\frac{z_t^2}{2\sigma_{\alpha_0}^2}\bigg\}dz_t\nonumber\\
&\geq\sum_{a'=1}^{k-2}\int_{a'\mu_\alpha}^{(a'+1)\mu_\alpha}\frac{1}{\sqrt{2\pi}\sigma_{\beta_0}}\exp\{-\frac{y_t^2}{2\sigma_{\beta_0}^2}\}dy_t
\sum_{b'=1}^{2k-2}\int_{b'\mu_\alpha}^{(b'+1)\mu_\alpha}\frac{1}{\sqrt{2\pi}\sigma_{\alpha_0}}\exp\{-\frac{z_t^2}{2\sigma_{\alpha_0}^2}\}dz_t\nonumber\\
&\to\sum_{a'=1}^{k-2}\frac{\mu_\alpha}{\sqrt{2\pi}\sigma_{\beta_0}}\exp\{-\frac{a'^2\mu_\alpha^2}{2\sigma_{\beta_0}^2}\}
\sum_{b'=1}^{2k-2}\frac{\mu_\alpha}{\sqrt{2\pi}\sigma_{\alpha_0}}\exp\{-\frac{b'^2\mu_\alpha^2}{2\sigma_{\alpha_0}^2}\},
\end{align}
where the first inequality is based on the definition of
$\{T_{1,\{a',b'\}},T_{2,\{a',b'\}}\}$, and the second one comes from
(\ref{sigma0}). Noting that $\mu_\alpha/\sigma_{\alpha_0}$ and
$\mu_\alpha/\sigma_{\beta_0}$ scale as $\Theta(\sqrt{r})$, while
$k\mu_\alpha/\sigma_{\alpha_0}$ and $k\mu_\alpha/\sigma_{\beta_0}$
go to $\infty$ as $n\to\infty$, the last line in (\ref{lbprob})
converges to
\begin{eqnarray*}
\int_{0}^{\infty}\frac{1}{\sqrt{2\pi}}\exp\{-\frac{x^2}{2}\}dx
\int_{0}^{\infty}\frac{1}{\sqrt{2\pi}}\exp\{-\frac{y^2}{2}\}dy=1/4,
\end{eqnarray*}
which concludes the proof.


\subsection{LADA-U Algorithm}\label{appLADAU}
In this appendix, we introduce the LADA-U (Uniform) algorithm, which
achieves the goal of distributed averaging by simulating a
nonreversible chain with uniform stationary distribution on the
geometric random graph. In LADA-U, each node $i$ holds four values
$y_i^l$, $l=0,\cdots,3$ corresponding to the four directions, all
initialized to $x_i(0)$. During each iteration, the east value of
node $i$ is updated with
\begin{eqnarray*}
y_i^{0}(t+1)&=&(1-p)\left[\sum_{j\in
\mathcal{N}_i^2\bigcup\widehat{N}_i^2}\frac{y_{j}^{0}(t)}{d_{\max}}+
\sum_{j\in
\widetilde{\mathcal{N}}_i^2}\frac{y_{j}^{2}(t)}{d_{\max}}+\left(1-\frac{d_i^2}{d_{\max}}\right)y_{i}^{2}(t)\right]\\
&&+\frac{1}{2}p\left(y_{i}^{1}(t)+y_{i}^{3}(t)\right)
\end{eqnarray*}
where $d_{\max}=\max_{i,l} d_i^l$, and $p=\Theta(r)$ is defined
similarly as in LADA. Note that the boundary effect have been
addressed through virtual neighbors as in LADA. The north, west and
south values are updated in the same fashion. Node $i$ computes its
estimate of $x_{\mathrm{ave}}$ with $
x_i(t+1)=\frac{1}{4}\sum_{l=0}^3y_i^l(t+1)$.


We then give some performance analysis for LADA-U. Denote
$\mathbf{y}$ as in LADA, the iteration can be written as
$\mathbf{y}(t+1)=\tilde{\mathbf{P}}_2^T\mathbf{y}(t)$, where
$\tilde{\mathbf{P}}_2$ is a doubly stochastic matrix through our
design. The exchange weights for an east value of some node $i$ are
illustrated in Fig. \ref{LADA-U}: a fraction $\frac{p}{2}$ of the
east value goes to the north and south value of the same node
respectively, a total fraction of $\frac{d_i^0}{d_{\max}}(1-p)$ goes
uniformly to the east values of $d_i^0$  east neighbors, and the
remaining $\left(1-\frac{d_i^0}{d_{\max}}\right)(1-p)$ goes to the
west value of node $i$. The transitions between the east and west
state make up for the difference in $d_i^0$ and $d_i^2$, and ensures
that the incoming probabilities for each state also sum to 1. While
such a design guarantees that the associated chain has a uniform
stationary distribution, it also introduces some diffusive behavior,
hence the centralized performance can only be achieved with a larger
$r$. In the following, we show that for LADA-U,
$T_{\mathrm{ave}}(\epsilon)=O(r^{-1}\log(\epsilon^{-1}))$ when the
transmission radius $r=\Omega\left(\left(\frac{\log
n}{n}\right)^{\frac{1}{3}}\right)$ with high probability.

\begin{figure} \centering
\includegraphics[width=3.0in, bb=190 320 490
500]{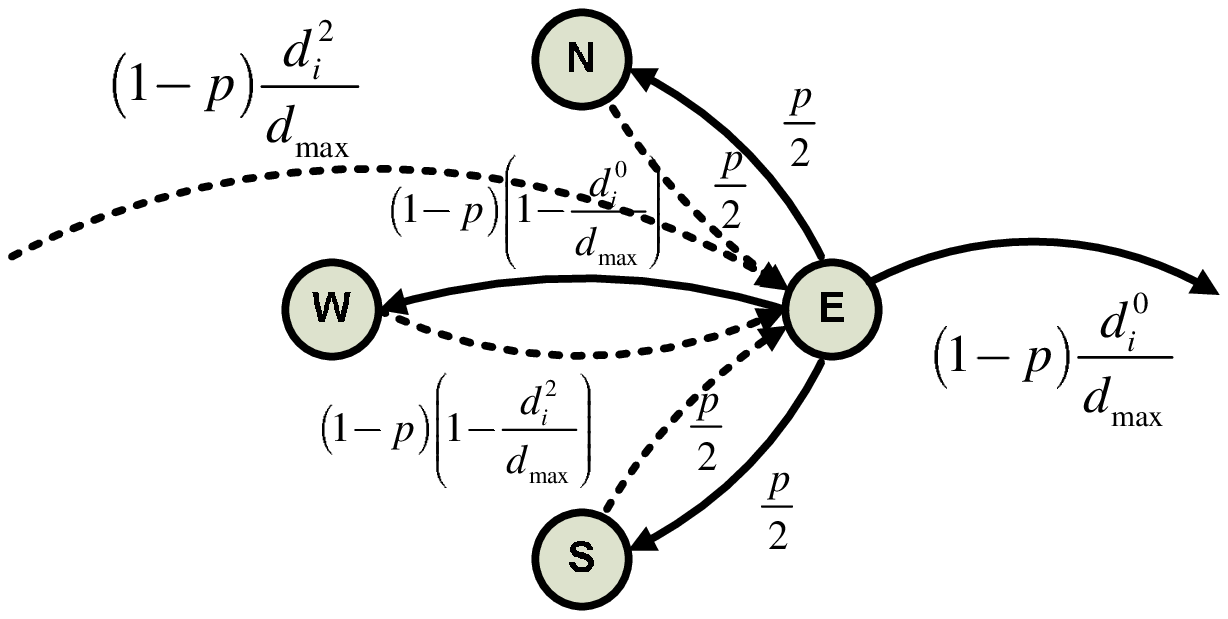}\caption{The Markov chain used in LADA-U: outgoing
probabilities (solid lines) and incoming probabilities (dotted
lines) for the east state are depicted}\label{LADA-U}
\end{figure}


It can be shown that the expected location of the random walk
$\tilde{\mathbf{P}}_2$ evolves according to a random walk
$\tilde{\mathbf{P}}'$ on the $k\times k$ grid, where
$k=\frac{1}{\mu_\alpha}+1$ as defined in Appendix \ref{appLADA}.
$\tilde{\mathbf{P}}'$ differs from $\tilde{\mathbf{P}}$ used in
Section IV in two aspects: 1) there are additional probabilities of
moving between states of opposite directions corresponding to the
same node; 2) a 90 degree turn is towards a state corresponding to
the same node instead of the next node in the turning direction.
Recall that from Lemma \ref{regularity}, when
$r=\Omega\left(\left(\frac{\log n}{n}\right)^{\frac{1}{3}}\right)$,
we have $d_i^l=\frac{n\pi r^2}{4}(1\pm O(1))$ for all $i$ and $l$
w.h.p. Thus for each move, the probability that the random walk
$\tilde{\mathbf{P}}'$ keeps the direction is at least
$(1-p)\frac{d_{\min}}{{d_{\max}}}=(1-1/k)(1-O(r))>1-\frac{c_1}{k}$
for some constant $c_1>1$ w.h.p. During the first $6k$ moves, the
probability that the random walk $\tilde{\mathbf{P}}'$ makes exactly
two 90 degree turns towards given directions at given times $T_1$
and $T_2$, and keeps direction for the remaining moves is at least
$\frac{1}{4k^2}\left(1-\frac{c_1}{k}\right)^{6k-2}\geq
\frac{2^{-12c_1}}{4k^2}$. Then, following the argument in Appendix
\ref{appgrid}, if the random walk $\tilde{\mathbf{P}}'$ starts from
an east or west state, any east or west state can be reached with
probability at least $\frac{2^{-12c_1}}{4k^2}$ in $6k$ steps (note
that the modification in the 90 degree turns only causes constant
shifts in the expressions of $s_t$, and does not affect the result).
The case for north and south states can be similarly argued, and we
conclude that the state distribution of the random walk
$\tilde{\mathbf{P}}'$ is approximately uniform at $t=6k$ w.h.p.
Then, following the analysis in Appendix \ref{appLADA}, it can be
shown that the exact location of random walk $\tilde{\mathbf{P}}_2$
is also approximately uniform at $t=6k$, which by the uniformity of
the stationary distribution of $\tilde{\mathbf{P}}_2$ implies that
the $\epsilon$-mixing time of $\tilde{\mathbf{P}}_2$, as well as the
$\epsilon$-averaging time of LADA-U is
$O(r^{-1}\log(\epsilon^{-1}))$ w.h.p.

\subsection{Distributed Clustering}
We assume each node $i$ has an initial seed $s_i$ which is unique
within its neighborhood. This can be realized through, e.g., drawing
a random number from a large common pool, or simply using nodes'
IDs. From time 0, each node $i$ starts a timer with length
$t_i=s_i$, which is decremented by 1 at each time instant as long as
it is greater than 0. If node $i$'s timer expires (reaches 0), it
becomes a cluster-head, and broadcasts a ``cluster\_initialize"
message to all its neighbors. Each of its neighbors with a timer
greater than 0 signals its intention to join the cluster by replying
with a ``cluster\_join" message, and also sets the timer to 0. If a
node receives more than one ``cluster\_initialize" messages at the
same time, it randomly chooses one cluster-head and replies with the
``cluster\_join" message. At the end, clusters are formed such that
every node belongs to one and only one cluster. The uniqueness of
seeds within the neighborhood ensures that cluster-heads are at
least of distance $r$ from each other. We assume that clusters are
formed in advance and the overhead is amortized over the multiple
computations. The detailed algorithm is given in Algorithm 4.

\begin{algorithm}[H]
\caption{Distributed Clustering}
\begin{algorithmic}
\STATE $K\Leftarrow 0$ \COMMENT{$K$: number of clusters}
\FORALL{$i\in V$} \STATE $t_i \Leftarrow s_i$ \ENDFOR \REPEAT
\FORALL{$i$ with $t_i>0$} \STATE $t_i\Leftarrow t_i-1$ \IF{$t_i=0$}
\STATE
$K\Leftarrow K+1$, $C_K\Leftarrow \{i\}$ \COMMENT{$C_k$: nodes in cluster $k$}
\FORALL{$j\in \mathcal{N}_i$ and with $t_j>0$} \STATE $t_j\Leftarrow
0$, $C_K\Leftarrow C_K\bigcup\{j\}$ \ENDFOR \ENDIF \ENDFOR
\UNTIL{$\bigcup_k C_k=V$}
\end{algorithmic}
\end{algorithm}


\end{document}